\documentclass[journal, onecolumn, pagenumbers]{IEEEtran}
\usepackage{listings}
\usepackage{amsmath}
\usepackage{inputenc}
\usepackage{amsthm}
\usepackage{tikz}
\usepackage{caption}
\usepackage{array}
\usepackage{mdwmath}
\usepackage{multirow}
\usepackage{mdwtab}
\usepackage{eqparbox}
\usepackage{amsfonts}
\usepackage{tikz}
\usepackage{multirow,bigstrut,threeparttable}
\usepackage{amsthm}
\usepackage{array}
\usepackage{bbm}
\usepackage{subfigure}
\usepackage{epstopdf}
\usepackage{mdwmath}
\usepackage{mdwtab}
\usepackage{eqparbox}
\usepackage{tikz}
\usepackage{latexsym}
\usepackage{cite}
\usepackage{amssymb}
\usepackage{bm}
\usepackage{amssymb}
\usepackage{graphicx}
\usepackage{mathrsfs}
\usepackage{epsfig}
\usepackage{psfrag}
\usepackage{setspace}
\usepackage{hyperref}
\usepackage{algorithm}
\usepackage{algpseudocode}
\usepackage{stfloats}

\newtheorem{conjecture}{Conjecture}

\newtheorem{theorem}{Theorem}

\newtheorem{lemma}{Lemma} 
\newtheorem{corollary}{Corollary}

\def \bP {\mathbb{P}}
\def \bE {\mathbb{E}}

\begin{document}

\title{Adaptive Estimation of Shannon Entropy}

\author{Yanjun~Han,~\IEEEmembership{Student Member,~IEEE},~Jiantao~Jiao,~\IEEEmembership{Member,~IEEE}, and Tsachy~Weissman,~\IEEEmembership{Fellow,~IEEE}
\thanks{Yanjun Han and Tsachy Weissman are with the Department of Electrical Engineering, Stanford University, CA, USA. Email: \{yjhan, tsachy\}@stanford.edu Jiantao Jiao is with the Department of Electrical Engineering and Computer Sciences, University of California, Berkeley, CA, USA. Email: jiantao@eecs.berkeley.edu}. 
\thanks{The materials in this paper was presented in part at the 2015 IEEE International Symposium on Information Theory, Hong Kong, China. }
}
%
%

\date{\today}

\vspace{-10pt}

\maketitle

\begin{abstract}
We consider estimating the Shannon entropy of a discrete distribution $P$ from $n$ i.i.d. samples. Recently, Jiao, Venkat, Han, and Weissman, and Wu and Yang constructed approximation theoretic estimators that achieve the minimax $L_2$ rates in estimating entropy. Their estimators are consistent given $n \gg \frac{S}{\ln S}$ samples, where $S$ is the support size, and it is the best possible sample complexity. In contrast, the Maximum Likelihood Estimator (MLE), which is the empirical entropy, requires $n\gg S$ samples.

In the present paper we significantly refine the minimax results of existing work. To alleviate the pessimism of minimaxity, we adopt the adaptive estimation framework, and show that the minimax rate-optimal estimator in Jiao, Venkat, Han, and Weissman achieves the minimax rates simultaneously over a nested sequence of subsets of distributions $P$, without knowing the support size $S$ or which subset $P$ lies in. In other words, their estimator is adaptive with respect to this nested sequence of the parameter space, which is characterized by the entropy of the distribution. We also characterize the maximum risk of the MLE over this nested sequence, and show, for every subset in the sequence, that the performance of the minimax rate-optimal estimator with $n$ samples is essentially that of the MLE with $n\ln n$ samples, thereby further substantiating the generality of the \emph{effective sample size enlargement} phenomenon identified by Jiao, Venkat, Han, and Weissman.
\end{abstract}


\section{Introduction}

Shannon entropy $H(P)$, defined as
\begin{equation}\label{eqn.entropy}
H(P) \triangleq \sum_{i = 1}^S p_i \ln \frac{1}{p_i},
\end{equation}
is one of the most fundamental quantities of information theory and statistics, which emerged in Shannon's 1948 masterpiece~\cite{Shannon1948} as the answer to foundational questions of compression and communication.

Consider the problem of estimating Shannon entropy $H(P)$ from $n$ i.i.d. samples. Classical theory is mainly concerned with the case where the number of samples $n\to \infty$, while the support size $S$ is fixed. In that scenario, the maximum likelihood estimator (MLE), $H(P_n)$, which plugs in the empirical distribution into the definition of entropy, is \emph{asymptotically efficient} \cite[Thm. 8.11, Lemma 8.14]{Vandervaart2000} in the sense of the H\'ajek convolution theorem \cite{Hajek1970characterization} and the H\'ajek--Le Cam local asymptotic minimax theorem \cite{Hajek1972local}. It is therefore not surprising to encounter the following quote from the introduction of Wyner and Foster \cite{Wyner--Foster2003lower} who considered entropy estimation:
\begin{center}
\parbox{.85\textwidth}{~~\emph{  ``The plug-in estimate is universal and optimal not only for finite alphabet i.i.d. sources but also for finite alphabet, finite memory sources. On the other hand, practically as well as theoretically, these problems are of little interest.
"}}
\end{center}

In contrast, various modern data-analytic applications deal with datasets which do not fall into the regime of fixed alphabet and $n\to \infty$. In fact, in many applications the support size $S$ is comparable to, or even larger than the number of samples $n$.

For example:
\begin{itemize}
	\item Corpus linguistics: about half of the words in the Shakespearean canon appeared only once \cite{Efron--Thisted1976}.
	\item Network traffic analysis: many customers or website users are seen a small number of times \cite{Benevenuto--Rodrigues--Cha--Almeida2009characterizing}.
	\item  Analyzing neural spike trains: natural stimuli generate neural responses of high timing precision resulting in a massive space of meaningful responses \cite{Berry--Warland--Meister1997structure,Mainen--Sejnowski1995reliability,Van--Lewen--Strong--Koberle--Bialek1997reproducibility}.
\end{itemize}

\subsection{Existing literature}

The problem of entropy estimation in the large alphabet regime (or non-asymptotic analysis) has been investigated extensively in various disciplines, which we refer to~\cite{jiao2014minimax} for a detailed review. One recent breakthrough in this direction came from Valiant and Valiant~\cite{Valiant--Valiant2011}, who constructed the first explicit entropy estimator whose sample complexity is $n \asymp \frac{S}{\ln S}$ samples, which they also proved to be necessary. It was also shown in \cite{Paninski2003}\cite{Jiao--Venkat--Han--Weissman2014MLE} that the MLE requires $n \asymp S$ samples, implying that MLE is strictly sub-optimal in terms of sample complexity.

However, the aforementioned estimators have not been shown to achieve the minimax $L_2$ rates. In light of this, Wu and Yang~\cite{Wu--Yang2014minimax} and Jiao et al.\cite{jiao2014minimax} independently developed schemes based on approximation theory that achieved the minimax $L_2$ convergence rates for the entropy. Furthermore, Jiao et al.\cite{jiao2014minimax} proposed a general methodology for estimating functionals, and showed that for a wide class of functionals (including entropy, mutual information, and power sum functionals), their methodology leads to minimax rate-optimal estimators whose performance with $n$ samples is essentially that of the MLE with $n\ln n$ samples. The approximation ideas proved to be very fruitful in Acharya \emph{et al.}~\cite{Acharya--Orlitsky--Suresh--Tyagi2014complexity}, Wu and Yang~\cite{wu2015chebyshev}, Han, Jiao, and Weissman~\cite{Han--Jiao--Weissman2016minimaxdivergence}, Jiao, Han, and Weissman~\cite{jiao2016minimaxl1distance}, Bu \emph{et al.}~\cite{bu2016estimation}, Orlitsky, Suresh, and Wu~\cite{orlitsky2016optimal}, Wu and Yang~\cite{wu2016sample}.

On the practical side, Jiao et al.\cite{Jiao--Venkat--Han--Weissman2014beyond} showed that the minimax rate-optimal estimators introduced in~\cite{jiao2014minimax} can lead to consistent and substantial performance boosts in various machine learning algorithms.

Recall that the minimax risk of estimating functional $F(P)$ is defined via $\inf_{\hat{F}}\sup_{P \in \mathcal{M}_S} \bE_P \left( \hat{F} - F(P)\right)^2$, where $\mathcal{M}_S$ denotes all distributions with support size $S$, and the infimum is taken with respect to all estimators $\hat{F}$. Correspondingly, the maximum $L_2$ risk of MLE $F(P_n)$, which evaluates the functional $F(\cdot)$ at the empirical distribution $P_n$, is defined via $\sup_{P \in \mathcal{M}_S} \bE_P \left( F(P_n) -F(P) \right)^2$. The following table in Jiao et al.~\cite{jiao2014minimax} summaries the minimax $L_2$ rates and the $L_2$ rates of MLE in estimating $H(P)$ and $F_\alpha(P) \triangleq \sum_{i = 1}^S p_i^\alpha$. Whenever there are two terms, the first term corresponds to squared bias, and the second term corresponds to variance. It is evident that one can obtain the minimax rates from the $L_2$ rates of MLE via replacing $n$ with $n\ln n$ in the dominating (bias) terms. We adopt the following notation: notation $a_n\ll b_n$ or $a_n=o(b_n)$ means that $\limsup_{n\to\infty} a_n/b_n = 0$, and $a_n\gg b_n$ means $b_n\ll a_n$. Notation $a_n\lesssim b_n$ or $a_n=O(b_n)$ means $\sup_n a_n/b_n <\infty$, $a_n \gtrsim b_n$ means $b_n\lesssim a_n$, $a_n \asymp b_n$ or $a_n=\Theta(b_n)$ means $a_n\lesssim b_n$ and $a_n \gtrsim b_n$, or equivalently, there exists two universal constants $c,C$ such that
\begin{align}
  0<c<\liminf_{n\to\infty} \frac{a_n}{b_n} \le \limsup_{n\to\infty} \frac{a_n}{b_n} < C<\infty.
\end{align}

\begin{table}[h]
\begin{center}
    \begin{tabular}{| l | l | l |}
    \hline
    & Minimax squared error rates & Worst squared error rates of MLE\\ \hline
$H(P)$    & $\frac{S^2}{(n \ln n)^2} + \frac{\ln^2 S}{n} \quad \left( n \gtrsim \frac{S}{\ln S} \right)$ (\cite{jiao2014minimax}, \cite{Wu--Yang2014minimax}) & $\frac{S^2}{n^2} + \frac{\ln^2 S}{n}\quad \left( n\gtrsim S \right)$ \cite{Jiao--Venkat--Han--Weissman2014MLE}  \\ \hline

$F_\alpha(P), 0<\alpha\leq \frac{1}{2}$ &  $\frac{S^2}{(n \ln n)^{2\alpha}} \quad \left( n \gtrsim S^{1/\alpha}/\ln S, \ln n \lesssim \ln S \right)$ (\cite{jiao2014minimax})  & $\frac{S^2}{n^{2\alpha}}\quad \left( n \gtrsim S^{1/\alpha} \right)$ \cite{Jiao--Venkat--Han--Weissman2014MLE}     \\
    \hline

$F_\alpha(P), \frac{1}{2}<\alpha<1$ &    $\frac{S^2}{(n \ln n)^{2\alpha}} + \frac{S^{2-2\alpha}}{n}\quad \left( n \gtrsim S^{1/\alpha}/\ln S \right)$ (\cite{jiao2014minimax})  & $\frac{S^2}{n^{2\alpha}} + \frac{S^{2-2\alpha}}{n}\quad \left( n \gtrsim S^{1/\alpha} \right)$ \cite{Jiao--Venkat--Han--Weissman2014MLE}  \\ \hline

$F_\alpha(P), 1< \alpha<\frac{3}{2}$ & $(n \ln n)^{-2(\alpha-1)}\quad \left(S \gtrsim n\ln n \right)$ (\cite{jiao2014minimax})  & $n^{-2(\alpha-1)}\quad \left(S \gtrsim n \right)$ \cite{Jiao--Venkat--Han--Weissman2014MLE}  \\ \hline

$F_\alpha(P), \alpha\geq \frac{3}{2}$ & $n^{-1}$ \cite{Jiao--Venkat--Han--Weissman2014MLE}  & $n^{-1}$ \cite{Jiao--Venkat--Han--Weissman2014MLE}  \\ \hline
    \end{tabular}

    \caption{Comparison of the minimax $L_2$ rates and the $L_2$ rates of MLE in estimating $H(P)$ and $F_\alpha(P)\triangleq \sum_{i = 1}^S p_i^\alpha$. Whenever there are two terms, the first term corresponds to squared bias, and the second term corresponds to variance. It is evident that one can obtain the minimax rates from the $L_2$ rates of MLE via replacing $n$ with $n\ln n$ in the dominating (bias) terms. }
     \label{table.summary}
\end{center}
\end{table}

\subsection{Refined minimaxity: adaptive estimation}

One concern the readers may have about results on minimax rates is that they are too pessimistic. Indeed, in the definition $\inf_{\hat{F}}\sup_{P \in \mathcal{M}_S} \bE_P \left( \hat{F} - F(P)\right)^2$, we have considered the worst case distribution $P$ over all possible distributions supported on $S$ elements, and it would be disappointing if the estimator in Jiao et al.~\cite{jiao2014minimax} turned out to behave sub-optimally when we consider distributions lying in subsets of $\mathcal{M}_S$. A usual approach to alleviate this concern is the adaptive estimation framework, which we briefly review below.

The primary approach to alleviate the pessimism of minimaxity in statistics is the construction of adaptive procedures, which has gained particular prominence in nonparametric statistics~\cite{Cai2012minimax}. The goal of adaptive inference is to construct a single procedure that achieves optimality simultaneously over a collection of parameter spaces. Informally, an adaptive procedure automatically adjusts to the \emph{unknown} parameter, and acts as if it knows the parameter lies in a more restricted subset of the whole parameter space. A common way to evaluate such a procedure is to compare its maximum risk over each subset of the parameter space in the collection with the corresponding minimax risk. If they are nearly equal, then we say such a procedure is \emph{adaptive} with respect to that collection of subsets of the parameter space.

The primary results of this paper are twofold.
\begin{enumerate}
\item First, we show that the minimax rate-optimal entropy estimator in Jiao et al.~\cite{jiao2014minimax} is adaptive with respect to the collection of parameter space $\mathcal{M}_S(H)$, where $\mathcal{M}_S(H) \triangleq \{P: H(P)\leq H, P \in \mathcal{M}_S\}$. Moreover, the estimator does not need to know $S$ nor $H$, which is an advantage in practice since usually the support size $S$ nor an \emph{a priori} upper bound on the true entropy $H(P)$ are known.

\item Second, we show that the sample size \emph{enlargement} effect still holds in this adaptive estimation scenario. Table~\ref{table.summary} demonstrates that in estimating various functionals, the performance of the minimax rate-optimal estimator with $n$ samples is nearly that of the MLE with $n\ln n$ samples, which the authors termed ``effective sample size enlargement'' in \cite{jiao2014minimax}. We compute the maximum risk of the MLE over each $\mathcal{M}_S(H)$, and show that for every $H$, the performance of the estimator in~\cite{jiao2014minimax} with $n$ samples is still nearly that of the MLE with $n\ln n$ samples.
\end{enumerate}

These facts suggest that the estimator in Jiao et al.~\cite{jiao2014minimax} is near \emph{optimal} in a very strong sense, for which we refer the readers to \cite{jiao2014minimax} for a detailed discussion on methodology behind their estimator, literature survey, and experimental results.

\subsection{Mathematical framework and estimator construction}

Before we discuss the main results, we would like to recall the construction of the entropy estimator in~\cite{jiao2014minimax}. The approach is to tackle the estimation problem separately for the cases of ``small $p$'' and ``large $p$'' in $H(P)$ estimation, corresponding to treating regions where the functional is ``nonsmooth'' and ``smooth'' in different ways. Specifically, after we obtain the empirical distribution $P_n$, for each coordinate $P_n(i)$, if $P_n(i) \ll \ln n/n$, we (i) compute the best polynomial approximation for $-p_i \ln p_i$ in the regime $0\leq p_i \ll \ln n/n$, (ii) use the unbiased estimators for integer powers $p_i^k$ to estimate the corresponding terms in the polynomial approximation for $-p_i \ln p_i$ up to order $K_n \sim \ln n$, and (iii) use that polynomial as an estimate for $-p_i \ln p_i$. If $P_n(i) \gg \ln n/n$, we use the estimator $-P_n(i) \ln P_n(i) + \frac{1}{2n}$ to estimate $-p_i \ln p_i$. Then, we add the estimators corresponding to each coordinate.

We define the minimax risk for Multinomial model with $n$ observations on support size $S$ for estimating $H(P), P\in \mathcal{M}_S(H)$ as
\begin{equation}
R(S,n,H) \triangleq \inf_{\hat{H}} \sup_{P \in \mathcal{M}_S(H)} \bE_{\mathrm{Multinomial}} \left( \hat{H} - H(P) \right)^2,
\end{equation}
which is the quantity we will characterize in this paper. To simplify the analysis, we also utilize the Poisson sampling model, i.e., we first draw a random variable $N \sim \mathsf{Poi}(n)$, and then obtain $N$ samples from the distribution $P$. It is equivalent to having a $S$-dimensional random vector $\mathbf{Z}$ such that each component $Z_i$ in $\mathbf{Z}$ has distribution $\mathsf{Poi}(np_i)$, and all coordinates of $\mathbf{Z}$ are independent.

The counterpart of minimax risk in the Poissonized model is defined as
\begin{equation}
R_P(S,n,H) \triangleq \inf_{\hat{H}} \sup_{P \in \mathcal{M}_S(H)} \bE_{\mathrm{Poisson}} \left( \hat{H} - H(P) \right)^2.
\end{equation}

The following lemma, which follows from \cite{Wu--Yang2014minimax,jiao2014minimax}, shows that the minimax risks under the Multinomial model and the Poissonized model are essentially equivalent.
\begin{lemma}\label{lemma.poissonmultinomial}
The minimax risks under the Poissonized model and the Multinomial model are related via the following inequalities:
\begin{equation}
R_P(S,2n,H) - e^{-n/4} H^2 \leq  R(S,n,H) \leq 2 R_P(S,n/2,H).
\end{equation}
\end{lemma}

For simplicity of analysis, we conduct the classical ``splitting'' operation \cite{Tsybakov2013aggregation} on the Poisson random vector $\mathbf{Z}$, and obtain two independent identically distributed random vectors $\mathbf{X} = [X_1,X_2,\ldots,X_S]^T, \mathbf{Y} = [Y_1,Y_2,\ldots,Y_S]^T$, such that each component $X_i$ in $\mathbf{X}$ has distribution $\mathsf{Poi}(np_i/2)$, and all coordinates in $\mathbf{X}$ are independent. For each coordinate $i$, the splitting process generates a random variable $T_i$ such that $T_i|\mathbf{Z} \sim \mathrm{B}(Z_i, 1/2)$, and assign $X_i = T_i, Y_i = Z_i - T_i$. All the random variables $\{T_i:1\leq i\leq S\}$ are conditionally independent given our observation $\mathbf{Z}$. We also note that for random variable $X$ such that $nX \sim \mathsf{Poi}(np)$,
\begin{equation}
\bE \prod_{r = 0}^{k-1} \left( X - \frac{r}{n} \right) = p^k,
\end{equation}
for any $k \in \mathbb{N}_+$.

For simplicity, we re-define $n/2$ as $n$, and denote
\begin{equation}
\hat{p}_{i,1} = \frac{X_i}{n}, \hat{p}_{i,2} = \frac{Y_i}{n}, \Delta = \frac{c_1 \ln n}{n}, K = c_2 \ln n, t = \frac{\Delta}{4},
\end{equation}
where $c_1,c_2$ are positive parameters to be specified later. Note that $\Delta,K,t$ are functions of $n$, where we omit the subscript $n$ for brevity.

The estimator $\hat{H}$ in Jiao et al.~\cite{jiao2014minimax} is constructed as follows.
\begin{equation}\label{eqn.falphaconstruct}
\hat{H} \triangleq \sum_{i = 1}^S \left[  L_H(\hat{p}_{i,1}) \mathbbm{1}(\hat{p}_{i,2} \leq 2\Delta) + U_H(\hat{p}_{i,1}) \mathbbm{1}(\hat{p}_{i,2}>2\Delta) \right],
\end{equation}
where
\begin{align}
S_{K,H}(x) & \triangleq  \sum_{k = 1}^{K} g_{k,H} (4\Delta)^{-k + 1} \prod_{r = 0}^{k-1} \left(x-\frac{r}{n}\right) \label{eqn.SKalpha}\\
L_H(x) & \triangleq \min \left\{ S_{K,H}(x) , 1 \right\}\label{eqn.Lalphadef} \\
U_H(x) & \triangleq I_n(x)\left( -x \ln x + \frac{1}{2n}\right). \label{eqn.Ualphadef}
\end{align}

We explain each equation in detail as follows.

\begin{enumerate}
\item Equation~(\ref{eqn.falphaconstruct}):
Note that $\hat{p}_{i,1}$ and $\hat{p}_{i,2}$ are i.i.d. random variables such that $n \hat{p}_{i,1} \sim \mathsf{Poi}(n p_i)$. We use $\hat{p}_{i,2}$ to determine whether we are operating in the ``nonsmooth'' regime or not. If $\hat{p}_{i,2} \leq 2\Delta$, we declare we are in the ``nonsmooth'' regime, and plug in $\hat{p}_{i,1}$ into function $L_H(\cdot)$. If $\hat{p}_{i,2}>2\Delta$, we declare we are in the ``smooth'' regime, and plug in $\hat{p}_{i,1}$ into $U_H(\cdot)$.

\item Equation~(\ref{eqn.SKalpha}):

The coefficients $r_{k,H}, 0\leq k \leq K$ are coefficients of the best polynomial approximation of $-x\ln x$ over $[0,1]$ up to degree $K$, i.e.,
\begin{equation}
\sum_{k = 0}^K r_{k,H} x^k = \arg \min_{y(x) \in \mathsf{poly}_{K}} \sup_{x\in [0,1]} |y(x)- (-x \ln x)|,
\end{equation}
where $\mathsf{poly}_{K}$ denotes the set of algebraic polynomials up to order $K$. Note that in general $g_{k,\alpha}$ depends on $K$, which we do not make explicit for brevity.

Then we define $\{g_{k,H}\}_{1\leq k\leq K}$
\begin{equation}\label{eqn.gkhdefine}
g_{k,H} = r_{k,H}, 2\leq k \leq K, g_{1,H} = r_{1,H} - \ln (4\Delta).
\end{equation}

Lemma~\ref{lem_small_p} shows that for $nX \sim \mathsf{Poi}(np)$,
\begin{equation}
\bE S_{K,H}(X) = \sum_{k = 1}^{K} g_{k,H} (4\Delta)^{-k + 1} p^k
\end{equation}
is a near-best polynomial approximation for $-p\ln p$ on $[0,4\Delta]$. Thus, we can understand $S_{K,H}(X), nX\sim \mathsf{Poi}(np)$ as a random variable whose expectation is nearly \footnote{Note that we have removed the constant term from the best polynomial approximation. It is to ensure that we assign zero to symbols we do not see. } the best approximation of function $-x\ln x$ over $[0,4\Delta]$.

\item Equation~(\ref{eqn.Lalphadef}):

Any reasonable estimator for $-p\ln p$ should be upper bounded by the value one. We cutoff $S_{K,H}(x)$ by upper bound $1$, and define the function $L_H(x)$, which means ``lower part''.

\item Equation~(\ref{eqn.Ualphadef}):

The function $U_H(x)$ (means ``upper part'') is nothing but a product of an interpolation function $I_n(x)$ and the bias-corrected MLE. The interpolation function $I_n(x)$ is defined as follows:
\begin{equation}
I_n(x) = \begin{cases} 0 & x \leq t \\
g\left( x-t; t \right) & t < x < 2t \\ 1 & x \geq  2t\end{cases}
\end{equation}

The following lemma characterizes the properties of the function $g(x;a)$ appearing in the definition of $I_n(x)$. In particular, it shows that $I_n(x) \in C^4[0,1]$.
\begin{lemma}\label{lemma.gxa}
For the function $g(x;a)$ on $[0,a]$ defined as follows,
\begin{equation}
g(x;a) \triangleq 126 \left( \frac{x}{a} \right)^5 - 420 \left(\frac{x}{a}\right)^6 + 540 \left( \frac{x}{a} \right)^7-315 \left( \frac{x}{a} \right)^8 + 70 \left( \frac{x}{a} \right)^9 ,
\end{equation}
we have the following properties:
\begin{align}
g(0;a) = 0,\quad & g^{(i)}(0;a) = 0, 1\leq i\leq 4 \\
g(a;a) = 1, \quad & g^{(i)}(a;a) = 0, 1\leq i\leq 4
\end{align}
\end{lemma}

The function $g(x;1)$ is depicted in Figure~\ref{fig.gx1}. \footnote{As pointed out in~\cite{Han--Jiao--Weissman2016minimaxdivergence}, it is not necessary to use the interpolation function to achieve the minimax rates. Here we keep it in order to be consistent with~\cite{jiao2014minimax}. }
\begin{center}
  \centering
  \centerline{\includegraphics[width=0.6\textwidth]{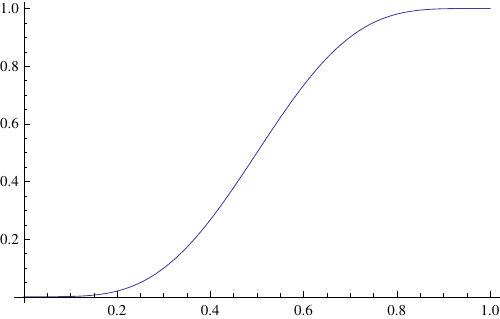}}
  \vspace{-0.1cm}
\captionof{figure}{The function $g(x;1)$ over interval $[0,1]$. }
\label{fig.gx1}
\end{center}
\end{enumerate}

\section{Main Results}
Since $\sup_{P\in\mathcal{M}_S}H(P) = \ln S$, we assume throughout this paper that $0<H\le\ln S$. Denote by $\mathcal{M}_S(H)$ the set of all discrete probability distributions $P$ with support size $|\mathsf{supp}(P)|=S$ and entropy $H(P)\le H$. We say an estimator $\hat{H}\equiv \hat{H}(\mathbf{Z})$ is within accuracy $\epsilon>0$, if and only if
\begin{align}
 \sup_{P\in\mathcal{M}_S(H)}   \left(\bE_P |\hat{H}-H(P)|^2\right)^{\frac{1}{2}} \le \epsilon.
\end{align}

For the plug-in estimator $H(P_n)$, the following theorem presents the non-asymptotic upper and lower bounds for the $L_2$ risk.

\begin{theorem}\label{th_MLE}
  If $H\ge H_0>0$, where $H_0$ is a universal positive constant, then for the plug-in estimator $H(P_n)$, we have
\begin{align}
  &\sup_{P\in\mathcal{M}_S(H)} \bE_P |H(P_n)-H(P)|^2 \\
  &\qquad\quad = \begin{cases}
  \Theta(1)\left[\left(\frac{S}{n}\right)^2 + \frac{H\ln S}{n}\right] & \text{if }S\ln S\le e^2nH,\\
  \left[\frac{H}{\ln S}\ln \left(\frac{S\ln S}{nH}\right) + O\left(\frac{H}{\ln S} + \frac{(\ln n)^2}{n}\right)\right]^2 & \text{otherwise}.
\end{cases}
\end{align}
\end{theorem}

Note that the only assumption in Theorem \ref{th_MLE} is that the upper bound $H$ should be no smaller than a constant, which is a reasonable assumption to avoid the subtle case where the naive zero estimator $\hat{H}\equiv 0$ has a satisfactory performance. The minimum sample complexity of the plug-in approach can be immediately obtained from Theorem \ref{th_MLE}.

\begin{corollary}\label{cor_1}
  If $H\ge H_0>0$, where $H_0$ is a universal positive constant, the plug-in estimator $H(P_n)$ is within accuracy $\epsilon$ if and only if
  \begin{align}
    n\gtrsim \begin{cases}
      S^{1-\frac{\epsilon}{H}}\cdot\frac{\ln S}{H} & \text{if } \frac{H}{\ln S} \ll \epsilon,\\
      \frac{S}{\epsilon} \vee \frac{H\ln S}{\epsilon^2} & \text{if } \frac{H}{\ln S} \gg \epsilon.
    \end{cases}
  \end{align}
\end{corollary}

Recall that it requires $n\gtrsim \frac{S}{\epsilon} \vee \frac{\ln^2 S}{\epsilon^2}$ samples for the MLE to achieve accuracy $\epsilon$ when there is no constraint on the entropy~\cite{jiao2014minimax,Wu--Yang2014minimax}. Hence, when the upper bound on the entropy is loose, i.e., $H\asymp\ln S$, the minimum sample complexity in the bounded entropy case is exactly the same, i.e., we cannot essentially improve the estimation performance. On the other hand, when the upper bound is tight, i.e., $H\ll\ln S$, the required sample complexity enjoyed a significant reduction. 

When it comes to the maximum $L_2$ risk, we conclude from Theorem \ref{th_MLE} that the bounded entropy property helps only at the boundary, i.e., when $n$ is close to $S$ and $H$ is small. Moreover, this help vanishes quickly as $S$ increases: when $n=S^{1-\delta}$, the maximum $L_2$ risk will be at the order $(\delta H)^2$, and the naive zero estimator achieves worst case risk $H^2$. 

Is the plug-in estimator $H(P_n)$ optimal in the minimax sense? It has been shown in \cite{Valiant--Valiant2011,jiao2014minimax,Wu--Yang2014minimax} that when there is no constraint on $H(P)$, i.e., $H=\ln S$, the answer is \emph{negative}. What about subsets of $\mathcal{M}_S$, such as $\mathcal{M}_S(H)$? The following theorem characterizes the minimax $L_2$ rates over $\mathcal{M}_S(H)$.
\begin{theorem}\label{th_minimax}
   If $H\ge H_0>0$, where $H_0$ is a universal positive constant, then
\begin{align}
  \inf_{\hat{H}}\sup_{P\in\mathcal{M}_S(H)} \bE_P |\hat{H}-H(P)|^2 \asymp
\begin{cases}
  \frac{S^2}{(n\ln n)^2} + \frac{H\ln S}{n} &\text{if }S\ln S\le e^2nH\ln n,\\
   \left[\frac{H}{\ln S}\ln \left(\frac{S\ln S}{nH\ln n}\right)\right]^2&\text{otherwise.}
\end{cases}
\end{align}
where the infimum is taken over all possible estimators. Moreover, the upper bound is achieved by the estimator $\hat{H}^*$ in~\cite{jiao2014minimax} under the Poissonized model without the knowledge of $H$ nor $S$, and in particular,
\begin{align}
  \sup_{P\in\mathcal{M}_S(H)} \bE_P |\hat{H}^*-H(P)|^2 \le \left[\frac{H}{\ln S}\ln \left(\frac{S\ln S}{nH\ln n}\right) + O\left(\frac{H}{\ln S} + \frac{(\ln n)^5}{n^{1-\epsilon}}\right)\right]^2
\end{align}
when $S\ln S> e^2nH\ln n$.
\end{theorem}

An immediate result on the sample complexity is as follows.
\begin{corollary}\label{cor_2}
 If $H\ge H_0>0$, where $H_0$ is a universal positive constant, the minimax rate-optimal estimator in~\cite{jiao2014minimax} is within accuracy $\epsilon$ if \begin{align}
    n\gtrsim \begin{cases}
      \frac{1}{H}S^{1-\frac{\epsilon}{H}} & \text{if } \frac{H}{\ln S} \ll \epsilon,\\
      \frac{S}{\epsilon\ln S} \vee \frac{H\ln S}{\epsilon^2} & \text{if } \frac{H}{\ln S} \gg \epsilon.
    \end{cases}
  \end{align}
\end{corollary}

For the minimum sample complexity, we still distinguish $H$ into two cases. Firstly, when $H\asymp \ln S$, the required sample complexity is $n \asymp \frac{S}{\epsilon\ln S}$, which recovers the minimax results with no constraint on entropy in \cite{jiao2014minimax}. Secondly, when $H\ll \ln S$, there is a significant improvement. 

\begin{conjecture}
We conjecture that the minimax rates in Theorem~\ref{th_minimax} can be refined to 
\begin{align}
  \inf_{\hat{H}}\sup_{P\in\mathcal{M}_S(H)} \bE_P |\hat{H}-H(P)|^2 =
\begin{cases}
 \Theta(1)\left( \frac{S^2}{(n\ln n)^2} + \frac{H\ln S}{n} \right) &\text{if }S\ln S\le e^2nH\ln n,\\
   \left[\frac{H}{2\ln S}\ln \left(\frac{S\ln S}{nH\ln n}\right)+ O\left(\frac{H}{\ln S} + \frac{(\ln n)^2}{n}\right)\right]^2&\text{otherwise.}
\end{cases}
\end{align}
In other words, we conjecture that the exact constant in the minimax rates in the regime $S \ln S > e^2 n H \ln n$ is $\frac{1}{2}$. It is partially justified by the observation that the minimax squared error without any samples is $(H/2)^2$. 
\end{conjecture}

We also conclude from Theorem \ref{th_minimax} that the bounded entropy constraint again helps only at the boundary, and this help vanishes quickly as $S$ increases: when $n=S^{1-\delta}$, we do not have sufficient information to make inference, and the naive zero estimator is near-minimax.

Moreover, it has been shown that the hard-thresholding estimator $P_s$ is an adaptive and near-minimax estimator of the discrete distribution $P$ given $H(P)\le H$ under $\ell_1$ loss \cite{Han--Jiao--Weissman2014minimax}. The next theorem shows that the plug-in estimator $H(P_s)$ is also far from optimal:

\begin{theorem}
  If $H\ge H_0>0$, where $H_0$ is a universal positive constant, then for the hard-thresholding estimator $P_s$ in \cite{Han--Jiao--Weissman2014minimax}, the plug-in estimator $H(P_s)$ satisfies
    \begin{align}
  &\sup_{P\in\mathcal{M}_S(H)} \bE_P |H(P_s)-H(P)|^2 \gtrsim \begin{cases}
  \left(\frac{S}{n}\right)^2 & \text{if }S\ln S\le e^2nH,\\
  \left[\frac{H}{\ln S}\ln \left(\frac{S\ln S}{nH}\right)\right]^2 & \text{otherwise}.
\end{cases}
\end{align}
\end{theorem}
\begin{proof}
  By definition we have $p_{s,i}\le p_{n,i}$ with strict inequality if and only if $p_{n,i}\le (\ln n)^\eta/n$ for some $\eta>1$. Hence, the monotone and concave property of $f(x)=-x\ln x$ on $[0,1/e]$ yields that, for sufficiently large $n$, $\bE H(P_s)\le \bE H(P_n)\le H(P)$, and thus
  \begin{align}
    \mathsf{Bias}(H(P_s)) &= |H(P) - \bE H(P_s)| \\
    &\ge |H(P) - \bE H(P_n)| = \mathsf{Bias}(H(P_n)).
  \end{align}
  The proof is completed by the proof of the lower bound in Theorem \ref{th_MLE} (cf. Section \ref{sec_lower_MLE}).
\end{proof}

To sum up, we have obtained the following conclusions.
\begin{enumerate}
\item The minimax rate-optimal entropy estimator in Jiao et al.~\cite{jiao2014minimax} is adaptive with respect to the collection of parameter space $\mathcal{M}_S(H)$, where $\mathcal{M}_S(H) \triangleq \{P: H(P)\leq H, P \in \mathcal{M}_S\}$. Moreover, the estimator does not need to know $S$ nor $H$, which is an advantage in practice since usually the support size $S$ nor an \emph{a priori} upper bound on the true entropy $H(P)$ are known.

\item Second, the sample size \emph{enlargement} effect still holds in this adaptive estimation scenario. Table~\ref{table.summary} demonstrates that in estimating various functionals, the performance of the minimax rate-optimal estimator with $n$ samples is essentially that of the MLE with $n\ln n$ samples, which the authors termed ``sample size enlargement'' in \cite{jiao2014minimax}. Theorems~\ref{th_MLE} and~\ref{th_minimax} show that over every $\mathcal{M}_S(H)$, the performance of the estimator in~\cite{jiao2014minimax} with $n$ samples is still essentially that of the MLE with $n\ln n$ samples.
\end{enumerate}

\section{Proof of Upper Bounds in Theorem \ref{th_MLE}}
First we consider the case where $S\ln S\le e^2nH$. For the bias, it has been shown in \cite{Paninski2003} that
\begin{align}
  \mathsf{Bias}(H(P_n)) \le \ln\left(1+\frac{S-1}{n}\right) \le \frac{S}{n}.
\end{align}
As for the variance, \cite{Jiao--Venkat--Han--Weissman2014MLE} shows that by the Efron-Stein inequality, we have
\begin{align}
  \mathsf{Var}(H(P_n)) \le \frac{2}{n}\sum_{i=1}^S p_i\left(\ln p_i-2\right)^2 \le \frac{2}{n}\left(\sum_{i=1}^S p_i(\ln p_i)^2 + 4H + 4\right).
\end{align}

\begin{lemma}\label{lem_varentropy}
  For any discrete distribution $P=(p_1,p_2,\cdots,p_S)$, we have
\begin{align}
  \sum_{i=1}^S p_i(\ln p_i)^2 \le 2\ln S\cdot\left(\sum_{i=1}^S -p_i\ln p_i\right) + 3.
\end{align}
\end{lemma}
In light of Lemma \ref{lem_varentropy}, we conclude that
\begin{align}
  \mathsf{Var}(H(P_n)) &\le\frac{2}{n}\left(\sum_{i=1}^S p_i(\ln p_i)^2 + 4H + 4\right)\\
  & \le \frac{2}{n}\left(2H\ln S + 4H + 7\right)\lesssim\frac{H\ln S}{n}
\end{align}
where we have used the assumption $H\ge H_0>0$ in the last step.

Hence, when $S\ln S\le e^2nH$, we have
\begin{align}
  \bE_P\left(H(P_n)-H(P)\right)^2 &= \left(\mathsf{Bias}(H(P_n))\right)^2 + \mathsf{Var}(H(P_n))\\
  & \lesssim \frac{S^2}{n^2} + \frac{H\ln S}{n}
\end{align}
which completes the proof for the first part. For the second part, we introduce a lemma first.

\begin{lemma}\label{lem_bias}
  Given $n\ge2$. For $p\le \frac{1}{n}$ and $n\hat{p}\sim\mathsf{B}(n,p)$, we have
  \begin{align}
    -p\ln(np) + \frac{\ln 2}{e}\cdot(n-1)p^2 &\le -p\ln p - \bE [-\hat{p}\ln\hat{p}] \\
    &\le -p\ln(np) + \ln2\cdot (n-1)p^2.
  \end{align}
\end{lemma}

Define $f(x)=-x\ln x$ on $[0,1]$. We also know that $\left|\mathsf{Bias}\left(f(\hat{p}_i)\right)\right|\le \frac{1}{n}$ holds for all $i$ \cite{Strukov--Timan1977mathematical}. In light of the previous result and Lemma \ref{lem_bias}, we have
\begin{align}
  \left|\mathsf{Bias}\left(H(P_n)\right)\right| &= \sum_{i: p_i\le \frac{1}{n}} \left|\mathsf{Bias}\left(f(\hat{p}_i)\right)\right|
  + \sum_{i: p_i> \frac{1}{n}} \left|\mathsf{Bias}\left(f(\hat{p}_i)\right)\right|\\
  &\le \sum_{i: p_i\le \frac{1}{n}} \left[-p_i\ln(np_i) + \ln2\cdot (n-1)p_i^2\right] + \sum_{i: p_i> \frac{1}{n}} \frac{1}{n}\\
  &\le \sum_{i: p_i\le \frac{1}{n}} -p_i\ln(np_i) +  \ln2 \cdot \sum_{i: p_i\le \frac{1}{n}} p_i + \sum_{i: p_i> \frac{1}{n}} \frac{1}{n} \\
  &\equiv B_1 + \ln2\cdot B_2 +  B_3.
\end{align}

Now we bound $B_1,B_2,B_3$ separately. By the concavity of $f(\cdot)$ we have
\begin{align}\label{eq:inverse}
  H_1 \triangleq \sum_{i: p_i\le \frac{1}{n}} -p_i\ln p_i \le -\left(\sum_{i: p_i\le \frac{1}{n}}p_i\right)\ln \left(\frac{1}{S_1}\sum_{i: p_i\le \frac{1}{n}}p_i\right)
\end{align}
where
\begin{align}
  S_1\triangleq \left|\left\{i: p_i\le \frac{1}{n}\right\}\right| \le S.
\end{align}

As a result of (\ref{eq:inverse}), we have
\begin{align}\label{eq:sum_lower_bound}
  \sum_{i: p_i\le \frac{1}{n}}p_i \ge S_1\cdot f^{-1}\left(\frac{H_1}{S_1}\right)
\end{align}
where $f^{-1}(\cdot)$ denotes the inverse of $f(\cdot)$, and we restrict $f^{-1}(\cdot)\in[0,e^{-1}]$ to avoid possible ambiguities. It is straightforward to verify that for any $a>0$ and $0<x_1\le x_2<1/(ea)$, we have
\begin{align}\label{eq:property_1}
  \frac{f^{-1}(ax_1)}{x_1} \le \frac{f^{-1}(ax_2)}{x_2}.
\end{align}
We summarize some more properties of $f^{-1}(\cdot)$ in the following lemma.

\begin{lemma}\label{lem_entropyfunction}
  For $0<x_1\le x_2<1/e$, we have
\begin{align}\label{eq:property_2}
  0 \le f^{-1}(x_2) - f^{-1}(x_1) \le \frac{x_2-x_1}{-\ln\left(-x_2/\ln(x_2)\right)-1}.
\end{align}
Moreover, for $y>e$,
\begin{align}\label{eq:property_3}
  f^{-1}(y^{-1}) = \frac{1}{y\ln(y\ln y)} + O\left(\frac{\ln\ln y}{y(\ln y)^3}\right).
\end{align}
\end{lemma}

Combining these properties of $f^{-1}(\cdot)$ yields
\begin{align}
  B_1 &= \sum_{i: p_i\le \frac{1}{n}} -p_i\ln(np_i)  = H_1 - \ln n\cdot \sum_{i: p_i\le \frac{1}{n}}p_i\\
  &\le H_1 - \ln n\cdot S_1f^{-1}\left(\frac{H_1}{S_1}\right) \label{eq:inequality1}\\
  &\le H_1 - \ln n\cdot Sf^{-1}\left(\frac{H_1}{S}\right)\label{eq:inequality2}\\
  &\le H - \ln n\cdot Sf^{-1}\left(\frac{H}{S}\right)\label{eq:inequality3}\\
  &= \frac{H}{\ln S}\ln \left[\frac{S\ln S}{nH}\right] + O\left(\frac{H}{\ln S}\cdot \frac{\ln\ln S}{\ln S}\right)
\end{align}
where (\ref{eq:inequality1}) follows from (\ref{eq:sum_lower_bound}), (\ref{eq:inequality2}) follows from (\ref{eq:property_1}), (\ref{eq:inequality3}) follows from (\ref{eq:property_2}) and
\begin{align}
  \frac{S}{eH}\ln \left(\frac{S}{H}\right) \ge \frac{S}{eH}\ln \left(\frac{S}{\ln S}\right) \ge \frac{S\ln S}{e^2H} \ge n
\end{align}
by assumption, and the last equality follows from (\ref{eq:property_3}).

Now we proceed to bound $B_2$, which is given by
\begin{align}
  B_2 = \sum_{i:p_i\le \frac{1}{n}}p_i \le \frac{1}{\ln n}\sum_{i:p_i\le \frac{1}{n}}-p_i\ln p_i \le \frac{H}{\ln n}.
\end{align}

As for $B_3$, due to the concavity of $f(\cdot)$, the minimum of $\sum_{i: p_i> \frac{1}{n}}f(p_i)$ is attained when all but one $p_i$ are at the boundary $p_i = \frac{1}{n}$, hence
\begin{align}
  H \ge \sum_{i=1}^S -p_i\ln p_i \ge \sum_{i: p_i> \frac{1}{n}}-p_i\ln p_i \ge \left(\left|\left\{i: p_i>\frac{1}{n}\right\}\right|-1\right)\cdot \frac{\ln n}{n}.
\end{align}
As a result, we have
\begin{align}
  B_3 = \frac{1}{n}\left|\left\{i: p_i>\frac{1}{n}\right\}\right| \le \frac{1}{n} + \frac{H}{\ln n}.
\end{align}

Hence,
\begin{align}
  \left|\mathsf{Bias}(H(P_n))\right| &\le B_1 + B_2 + B_3 \\
  &\le \frac{H}{\ln S}\ln\left[\frac{S\ln S}{nH}\right] + O\left(\frac{H}{\ln n}\right) + o(1)\\
  &\le \frac{H}{\ln S}\ln\left[\frac{S\ln S}{nH}\right] + O\left(\frac{H}{\ln S}\right) + o(1)
\end{align}
where the last inequality is obtained by separating two cases $S>n^\beta$ and $S\le n^\beta$ for some constant $\beta>1$, say, $\beta=2$. The proof is completed by noticing that \cite{Jiao--Venkat--Han--Weissman2014MLE}:
\begin{align}
  \mathsf{Var}(H(P_n)) \lesssim \frac{(\ln n)^2}{n}.
\end{align}

\section{Proof of Lower Bounds in Theorem \ref{th_MLE}}\label{sec_lower_MLE}
We first derive a lower bound for the bias term. We recall the following result in \cite{Jiao--Venkat--Han--Weissman2014MLE}.
\begin{lemma}\label{lem_bias_lower_MLE}
  For $p\ge \frac{15}{n},p\in[0,1]$, we have
\begin{align}
  -p\ln p-\bE[-\hat{p}\ln\hat{p}] \ge \frac{1-p}{2n} + \frac{1}{20n^2p} - \frac{p}{12n^2}.
\end{align}
\end{lemma}

If we choose
\begin{align}\label{eq:def_P0}
  P_0=\left(\frac{15}{n},\frac{15}{n},\cdots,\frac{15}{n},1-\frac{15K}{n},0,\cdots,0\right)\in\mathcal{M}_S
\end{align}
where
\begin{align}
  K = \min&\left\{\lfloor\frac{n}{15}\rfloor, S-1, \max\left\{N\in \mathbb{N}: -\frac{15N}{n}\ln\left(\frac{15}{n}\right) \right.\right.\\
  &\left.\left. \qquad\qquad - \left(1-\frac{15N}{n}\right)\ln \left(1-\frac{15N}{n}\right)\le H\right\}\right\}
\end{align}
we have
\begin{align}
  |\mathsf{Bias}(H(P_n))| &\ge K\cdot \left(\frac{n-15}{2n^2} + \frac{1}{300n} - \frac{5}{4n^3}\right) \\
  &\gtrsim
\frac{1}{n}\min\left\{n, S, \frac{nH}{\ln n}\right\} \gtrsim \frac{S}{n}
\end{align}
where we have used the assumption $S\ln S\le e^2nH$ and $H\ge H_0>0$ in the last inequality. Hence, we have proved that
\begin{align}
   \sup_{P\in\mathcal{M}_S(H)} \bE_P |H(P_n)-H(P)|^2 \gtrsim
      \left(\frac{S}{n}\right)^2 \qquad \text{if }S\ln S\le e^2nH. \label{eq:lower_bias_1}
\end{align}

For the lower bound in the case where $S\ln S> e^2nH$, we consider another distribution
\begin{align}
  P_1=\left(\frac{A}{S-1},\cdots,\frac{A}{S-1},1-A\right)\in\mathcal{M}_S
\end{align}
where $A\in(0,1)$ is the solution to the equation
\begin{align}\label{eq:def_A}
  -A\ln \left(\frac{A}{S-1}\right) - (1-A)\ln(1-A) = H.
\end{align}

From (\ref{eq:def_A}) it is easy to show that
\begin{align}
  \frac{H}{\ln S} \le A = \frac{H}{\ln S}\left(1 + O\left(\frac{\ln\ln S}{\ln S}\right)\right) = \frac{H}{\ln S}\left(1+o(1)\right)
\end{align}
then Lemma \ref{lem_bias} tells us that
\begin{align}
  |\mathsf{Bias}(H(P_n))| &\ge \sum_{i=1}^{S-1} -p_i\ln (np_i) = -A\ln\left(\frac{nA}{S-1}\right) \\
  & \ge \frac{H}{\ln S}\ln\left(\frac{nA}{S-1}\right)\\
   &= \frac{H}{\ln S}\ln\left(\frac{S\ln S}{nH}\right) + o\left(\frac{H}{\ln S}\right).\label{eq:lower_bias_2}
\end{align}

The lower bounds in Theorem \ref{th_MLE} for the bias part can thus be established by combining (\ref{eq:lower_bias_1}) and (\ref{eq:lower_bias_2}).

We now turn to the lower bound for variance. We will actually prove a stronger result: a minimax lower bound for all estimators for the $L_2$ risk, which naturally is also a lower bound for the maximum risk of the MLE. We use Le Cam's two-point method here. Suppose we observe a random vector ${\bf Z} \in (\mathcal{Z},\mathcal{A})$ which has distribution $P_\theta$ where $\theta \in \Theta$. Let $\theta_0$ and $\theta_1$ be two elements of $\Theta$. Let $\hat{T} = \hat{T}({\bf Z})$ be an arbitrary estimator of a function $T(\theta)$ based on $\bf Z$. We have the following general minimax lower bound.
\begin{lemma}\label{lem_twopoint}
  \cite[Sec. 2.4.2]{Tsybakov2008} Denoting the Kullback-Leibler divergence between $P$ and $Q$ by
  \begin{align}
    D(P\|Q) \triangleq \begin{cases}
      \int \ln\left(\frac{dP}{dQ}\right)dP, &\text{if }P\ll Q,\\
      +\infty, &\text{otherwise}.
    \end{cases}
  \end{align}
  we have
  \begin{align}
\inf_{\hat{T}} \sup_{\theta \in \Theta} \bP_\theta\left( |\hat{T} - T(\theta)| \geq \frac{|T(\theta_1)-T(\theta_0)|}{2} \right) \geq
\frac{1}{4}\exp\left(-D\left(P_{\theta_1}\|P_{\theta_0}\right)\right).
  \end{align}
\end{lemma}

Applying this lemma to the Poissonized model $n\hat{p}_i\sim\mathsf{Poi}(np_i),1\le i\le S$, we know that for $\theta_1=(p_1,p_2,\cdots,p_S), \theta_0=(q_1,q_2,\cdots,q_S)$,
\begin{align}
  D\left(P_{\theta_1}\|P_{\theta_0}\right) &= \sum_{i=1}^S D\left(\mathsf{Poi}(np_i)\|\mathsf{Poi}(nq_i)\right)\\
  &= \sum_{i=1}^S \sum_{k=0}^\infty \bP\left(\mathsf{Poi}(np_i)=k\right)\cdot k\ln\frac{p_i}{q_i} \\
  &= \sum_{i=1}^S np_i\ln\frac{p_i}{q_i} = nD(\theta_1\|\theta_0),
\end{align}
then for $\Delta=|H(\theta_1)-H(\theta_0)|$, Markov's inequality yields
\begin{align}
  \inf_{\hat{H}} \sup_{P \in \mathcal{M}_S(H)} \bE_P\left( \hat{H}-H(P)\right)^2 &\ge \frac{\Delta^2}{4}\cdot
  \inf_{\hat{H}} \sup_{P \in \mathcal{M}_S(H)} \bP\left( |\hat{H}-H(P)| \ge \frac{\Delta}{2}\right)\\
  &\ge \frac{\Delta^2}{16}\exp\left(-nD(\theta_1\|\theta_0)\right).
\end{align}

Fix $\epsilon\in(0,1)$ to be specified later, and let
\begin{align}
  \theta_1 &= \left(\frac{A}{S-1},\cdots,\frac{A}{S-1},1-A\right),\\
  \theta_0 &= \left(\frac{A(1-\epsilon)}{S-1},\cdots,\frac{A(1-\epsilon)}{S-1},1-A+A\epsilon\right),
\end{align}
where $A$ is the solution to (\ref{eq:def_A}). Direct computation yields
\begin{align}
  D(\theta_1\|\theta_0) = A\ln\frac{1}{1-\epsilon} + (1-A)\ln\frac{1-A}{1-A+A\epsilon} \triangleq h(\epsilon),
\end{align}
and it can be directly verified that $h(0)=h'(0)=0$, and $|h''(0)| = \frac{1-A}{A} >0$. Hence, for $\epsilon$ small enough we have $D(\theta_1\|\theta_0) \le \epsilon^2/A$. By choosing $\epsilon=(nA)^{-\frac{1}{2}}\lesssim1$, we have
\begin{align}
  \Delta &= |H(\theta_1)-H(\theta_0)| \\
  &= \left| A\ln(S-1) + H_b(A) - A(1-\epsilon)\ln(S-1) - H_b(A-A\epsilon)\right|\\
&\gtrsim A\epsilon\ln \left(\frac{S-1}{A}\right),
\end{align}
where
\begin{align}
  H_b(x) \triangleq -x\ln x-(1-x)\ln(1-x).
\end{align}

Hence, by Lemma \ref{lem_twopoint} and $nD(\theta_1\|\theta_0)\le1, A\asymp H/\ln S$ we can obtain the following minimax lower bound under the Poissonized model
\begin{align}\label{eq:var_lower}
  \inf_{\hat{H}}\sup_{P\in\mathcal{M}_S(H)} \bE_P |\hat{H}-H(P)|^2
  &\gtrsim \left[ A\epsilon\ln \left(\frac{S-1}{A}\right)\right]^2 \\
  &\asymp \frac{H}{n\ln S}\left[\ln\left(\frac{S\ln S}{H}\right)\right]^2\\
&\asymp \frac{H\ln S}{n}.
\end{align}

The corresponding minimax lower bound for the variance in the Multinomial model follows from Lemma \ref{lemma.poissonmultinomial}. The proof of Theorem \ref{th_MLE} is complete by combining the lower bounds for the bias and the variance.

\section{Proof of Upper Bounds in Theorem \ref{th_minimax}}
Define
\begin{align}
  \xi \triangleq \xi(X,Y) = L_H(X) \mathbbm{1}(Y \leq 2\Delta) + U_H(X) \mathbbm{1}(Y>2\Delta),
\end{align}
where $nX\overset{D}{=}nY\sim \mathsf{Poi}(np)$, and $X,Y$ are independent.
We first recall the following lemma from \cite{jiao2014minimax}.

\begin{lemma}\label{lem_optimal}
  Suppose $0<c_1=16(1+\delta), 0<8c_2\ln 2=\epsilon<1, \delta>0$. Then the bias and variance of $\xi(X,Y)$ are given as follows:
\begin{align}
  |\mathsf{Bias}(\xi)| &\lesssim \frac{1}{n\ln n}\\
  \mathsf{Var}(\xi) &\lesssim \frac{(\ln n)^4}{n^{2-\epsilon}} + \frac{p(\ln p)^2}{n}
\end{align}
\end{lemma}

In light of Lemma \ref{lem_optimal}, we have
\begin{align}
  |\mathsf{Bias}(\hat{H})| &\le \sum_{i=1}^S |\mathsf{Bias}(\xi(\hat{p}_{i,1},\hat{p}_{i,2}))| \lesssim \sum_{i=1}^S \frac{1}{n\ln n} = \frac{S}{n\ln n}\\
  \mathsf{Var}(\hat{H}) &= \sum_{i=1}^S \mathsf{Var}(\xi(\hat{p}_{i,1},\hat{p}_{i,2})) \le \sum_{i=1}^S \left(\frac{(\ln n)^4}{n^{2-\epsilon}} + \frac{p_i(\ln p_i)^2}{n}\right) \\
  &\lesssim \frac{S(\ln n)^4}{n^{2-\epsilon}} + \frac{H\ln S}{n}
\end{align}
where we have used Lemma \ref{lem_varentropy} in the last step. Hence,
\begin{align}
  \bE_P\left(\hat{H}-H(P)\right)^2 &= |\mathsf{Bias}(\hat{H})|^2 +  \mathsf{Var}(\hat{H}) \\
  &\lesssim \frac{S^2}{(n\ln n)^2} + \frac{S(\ln n)^4}{n^{2-\epsilon}} + \frac{H\ln S}{n}.
\end{align}

When $S\ln S\le enH\ln n$, for $\epsilon$ small enough, say, $\epsilon<\frac{1}{2}$, we have
\begin{align}
  \frac{S(\ln n)^4}{n^{2-\epsilon}}\lesssim \sqrt{\frac{S^2}{(n\ln n)^2}\cdot \frac{H\ln S}{n}} \le \frac{S^2}{(n\ln n)^2} + \frac{H\ln S}{n}
\end{align}
where we have used the assumption that $H\ge H_0>0$. Hence, the term $\frac{S(\ln n)^4}{n^{2-\epsilon}}$ is negligible when compared with others, and we have reached the end for the case $S\ln S\le e^2nH\ln n$.

For the case where $S\ln S\ge e^2nH\ln n$, we need stronger results for the bias and variance in the regime where $p<\frac{1}{en\ln n}$. The results are summarized in the following lemma.
\begin{lemma}\label{lem_small_p}
  If $0<c_2\le1\le c_1$, for $nX\sim\mathsf{Poi}(np),0<p<\frac{1}{en\ln n}$, we have
  \begin{align}
  |\bE S_{K,H}(X) + p\ln p| &\le -p\ln(pn\ln n) + \left(D_p + \ln(4c_1/c_2^2)\right)p\\
  \bE S_{K,H}^2(X) &\le 2^{10c_2\ln 2}\frac{(4c_1\ln n)^{4}p}{n}
\end{align}
  where the constant $D_p$ is given in Lemma \ref{lem_ibragimov}.
\end{lemma}

Using the Poisson tail bound (cf. Lemma \ref{lemma.poissontail}) and similar argument to \cite[Lemma 8]{jiao2014minimax}, we have the following lemma.
\begin{lemma}\label{lem_optimal}
  Suppose $0<c_1=16(1+\delta), 0<10c_2\ln 2=\epsilon<1, \delta>0$. Then for $0<p<\frac{1}{en\ln n}$, we have
\begin{align}
  |\mathsf{Bias}(\xi)| &\le -p\ln(pn\ln n) + c_3p\\
  \mathsf{Var}(\xi) &\lesssim \frac{(\ln n)^4p}{n^{1-\epsilon}}
\end{align}
where $c_3$ is some universal constant which only depends on $c_1$ and $c_2$.
\end{lemma}

Now we proceed to bound the total bias and variance. For the bias, we can write
\begin{align}
  \left|\mathsf{Bias}\left(\hat{H}\right)\right|
  &= \sum_{i: p_i\le \frac{1}{en\ln n}} \left|\mathsf{Bias}\left(\xi(\hat{p}_{i,1},\hat{p}_{i,2})\right)\right|
  + \sum_{i: p_i> \frac{1}{en\ln n}} \left|\mathsf{Bias}\left(\xi(\hat{p}_{i,1},\hat{p}_{i,2})\right)\right|\\
  &\le \sum_{i: p_i\le \frac{1}{en\ln n}} \left[-p_i\ln(p_in\ln n) + c_3p_i\right] + \sum_{i: p_i> \frac{1}{en\ln n}} \frac{O(1)}{n\ln n}\\
  &\le \sum_{i: p_i\le \frac{1}{en\ln n}} -p_i\ln(p_in\ln n) +  c_3 \sum_{i: p_i\le \frac{1}{en\ln n}} p_i + O(1)\cdot \sum_{i: p_i> \frac{1}{en\ln n}} \frac{1}{n\ln n} \\
  &\equiv B_1 + c_3B_2 +  O(1)\cdot B_3.
\end{align}

Using similar arguments in the proof of upper bound in Theorem \ref{th_MLE}, we can show that
\begin{align}
  B_1 &\le \frac{H}{\ln S}\ln\left(\frac{S\ln S}{nH\ln n}\right) + O\left(\frac{H}{\ln S}\cdot \frac{\ln\ln S}{\ln S}\right)\\
  B_2 &\le \frac{H}{\ln(en\ln n)} = O\left(\frac{H}{\ln n}\right)\\
  B_3 &\lesssim \frac{1}{n\ln n} + \frac{H}{\ln(en\ln n)} = O\left(\frac{H}{\ln n}\right).
\end{align}

Summing up the bias yields
\begin{align}
  |\mathsf{Bias}(\hat{H})| &\le B_1 + B_2 + B_3 \le \frac{H}{\ln S}\ln\left(\frac{S\ln S}{nH\ln n}\right) + O\left(\frac{H}{\ln S}\right).
\end{align}

As for the total variance, we have
\begin{align}
  \mathsf{Var}(\hat{H}) &= \sum_{i:p_i\ge \frac{1}{en\ln n}} \mathsf{Var}(\xi(\hat{p}_{i,1},\hat{p}_{i,2})) +  \sum_{i:p_i< \frac{1}{en\ln n}} \mathsf{Var}(\xi(\hat{p}_{i,1},\hat{p}_{i,2}))\\
 &\lesssim \sum_{i:p_i\ge \frac{1}{en\ln n}} \left(\frac{(\ln n)^4}{n^{2-\epsilon}} + \frac{p(\ln p)^2}{n}\right) + \sum_{i:p_i< \frac{1}{en\ln n}} \frac{(\ln n)^4p}{n^{1-\epsilon}}\\
 &\lesssim \left(\frac{(\ln n)^5}{n^{1-\epsilon}} + \frac{(\ln n)^2}{n}\right) + \frac{(\ln n)^4}{n^{1-\epsilon}} \le O\left(\frac{(\ln n)^5}{n^{1-\epsilon}}\right)
\end{align}
Combining the total bias and variance constitutes a complete proof of the upper bounds in Theorem \ref{th_minimax}.

\section{Proof of Lower Bounds in Theorem \ref{th_minimax}}
When $S\ln S\le e^2nH\ln n$, the lower bound for the squared bias, i.e., the $\frac{S^2}{(n\ln n)^2}$ term, can be obtained using a similar argument in \cite{Wu--Yang2014minimax}. Specifically, we can assign two product measures $\mu_0^N$ and $\mu_1^N$ to the first $N(\le S)$ components in the distribution vector $P$, where
\begin{align}
  \mathsf{supp}(\mu_i) = \{0\} \cup \left[\frac{1}{a_1n\ln n}, \frac{a_2\ln n}{n}\right], \qquad i=0,1
\end{align}
for some constants $a_1,a_2>0$, and
\begin{align}
  \int_0^1 t\mu_i(dt) = \frac{1}{a_1n\ln n},\qquad i=0,1.
\end{align}

In particular,
\begin{align}
  \int_0^1 -t\ln t\mu_1(dt) - \int_0^1 -t\ln t\mu_0(dt) \gtrsim \frac{1}{n\ln n}
\end{align}
and
\begin{align}
  \inf_{\hat{H}}\sup_{P\in\mathcal{M}_S} \bE_P\left(\hat{H}-H(P)\right)^2
   &\gtrsim \left[N\left(\int_0^1 -t\ln t\mu_1(dt) - \int_0^1 -t\ln t\mu_0(dt)\right)\right]^2\\
  & \gtrsim \frac{N^2}{(n\ln n)^2}.\label{eq:lower_minimax}
\end{align}

In \cite{Wu--Yang2014minimax}, $N=S$. However, in our case, we have an additional constraint that $H(P)\le H$. Since
\begin{align}
  \bE_{\mu_i}[-p\ln p] &= \int_0^1 -t\ln t\mu_i(dt)\\
  &\le \ln(a_1n\ln n)\int_0^1 t\mu_i(dt)\\
  & = \frac{a_1\ln(a_1n\ln n)}{n\ln n}\asymp \frac{1}{n}
\end{align}
we have
\begin{align}
  \bE_{\mu_i^N}H(P) = N\bE_{\mu_i}[-p\ln p] \asymp \frac{N}{n}.
\end{align}

One can show that the measures $\mu_i^N,i = 0,1$ are highly concentrated around their expectations~\cite{Wu--Yang2014minimax}. Hence, in order to ensure $H(P)\le H$ with overwhelming probability, we can set $N\asymp \min\{nH, S\}$, and the condition $S\ln S\le e^2nH\ln n$ and $H\ge H_0>0$ yield that $nH\gtrsim S$, and thus $N\gtrsim S$. Hence by (\ref{eq:lower_minimax}),
\begin{align}
  \inf_{\hat{H}}\sup_{P\in\mathcal{M}_S(H)} \bE_P\left(\hat{H}-H(P)\right)^2 \gtrsim  \frac{N^2}{(n\ln n)^2} \gtrsim \frac{S^2}{(n\ln n)^2}.
\end{align}

The variance bound $\frac{H\ln S}{n}$ has been given in (\ref{eq:var_lower}), and so far we have completed the proof of the first part. As for the second part, the key lemma we will employ is the so-called method of two fuzzy hypotheses presented in Tsybakov \cite{Tsybakov2008}. Below we briefly review this general minimax lower bound.

Suppose we observe a random vector ${\bf Z} \in (\mathcal{Z},\mathcal{A})$ which has distribution $P_\theta$ where $\theta \in \Theta$. Let $\sigma_0$ and $\sigma_1$ be two prior distributions supported on $\Theta$. Write $F_i$ for the marginal distribution of $\mathbf{Z}$ when the prior is $\sigma_i$ for $i = 0,1$. For any function $g$ we shall write $\bE_{F_i} g(\mathbf{Z})$ for the expectation of $g(\mathbf{Z})$ with respect to the marginal distribution of $\mathbf{Z}$ when the prior on $\theta$ is $\sigma_i$. We shall write $\bE_\theta g(\mathbf{Z})$ for the expectation of $g(\mathbf{Z})$ under $P_\theta$. Let $\hat{T} = \hat{T}({\bf Z})$ be an arbitrary estimator of a function $T(\theta)$ based on $\bf Z$. We have the following general minimax lower bound.

\begin{lemma}\cite[Thm. 2.15]{Tsybakov2008} \label{lemma.tsybakov}
Given the setting above, suppose there exist $\zeta\in \mathbb{R}, s>0, 0\leq \beta_0,\beta_1 <1$ such that
\begin{align}
\sigma_0(\theta: T(\theta) \leq \zeta -s) & \geq 1-\beta_0 \\
\sigma_1(\theta: T(\theta) \geq \zeta + s) & \geq 1-\beta_1.
\end{align}
If $V(F_1,F_0) \leq \eta<1$, then
\begin{equation}
\inf_{\hat{T}} \sup_{\theta \in \Theta} \bP_\theta\left( |\hat{T} - T(\theta)| \geq s \right) \geq \frac{1-\eta - \beta_0 - \beta_1}{2},
\end{equation}
where $F_i,i = 0,1$ are the marginal distributions of $\mathbf{Z}$ when the priors are $\sigma_i,i = 0,1$, respectively.
\end{lemma}

Here $V(P,Q)$ is the total variation distance between two probability measures $P,Q$ on the measurable space $(\mathcal{Z},\mathcal{A})$. Concretely, we have
\begin{equation}
V(P,Q) \triangleq \sup_{A\in \mathcal{A}} | P(A) - Q(A) | = \frac{1}{2} \int |p-q| d\nu,
\end{equation}
where $p = \frac{dP}{d\nu}, q = \frac{dQ}{d\nu}$, and $\nu$ is a dominating measure so that $P \ll \nu, Q \ll \nu$.

First we assume that $S\lesssim n^{\frac{3}{2}}$. In light of Lemma \ref{lemma.tsybakov}, we construct two measures as follows.
\begin{lemma}\label{lem_measure}
  For any $0<\eta<1$ and positive integer $L>0$, there exist two probability measures $\nu_0$ and $\nu_1$ on $[\eta,1]$ such that
  \begin{enumerate}
    \item $\int t^{l} \nu_1(dt) = \int t^{l} \nu_0(dt)$, for all $l=0,1,2,\cdots,L$;
    \item $\int -\ln t \nu_1(dt) - \int -\ln t \nu_0(dt) = 2E_L[-\ln x]_{[\eta,1]}$,
  \end{enumerate}
  where $E_L[-\ln x]_{[\eta,1]}$ is the distance in the uniform norm on $[\eta,1]$ from the function $g(x)=-\ln x$ to the space spanned by $\{1,x,\cdots,x^L\}$.
\end{lemma}

Based on Lemma \ref{lem_measure}, two new measures $\tilde{\nu}_0,\tilde{\nu}_1$ can be constructed as follows: for $i=0,1$, the restriction of $\tilde{\nu}_i$ on $[\eta,1]$ is absolutely continuous with respect to $\nu_i$, with the Radon-Nikodym derivative given by
\begin{align}\label{eq:derivative}
  \frac{d\tilde{\nu}_i}{d\nu_i}(t) = \frac{\eta}{t}, \qquad t\in[\eta,1],
\end{align}
and $\tilde{\nu}_i(\{0\})=1-\tilde{\nu}_i([\eta,1])\ge0$. Hence, $\tilde{\nu}_0,\tilde{\nu}_1$ are both probability measures on $[0,1]$, with the following properties
\begin{enumerate}
  \item $\int t^1 \tilde{\nu}_1(dt) = \int t^1 \tilde{\nu}_0(dt) = \eta$;
  \item $\int t^l \tilde{\nu}_1(dt) = \int t^l \tilde{\nu}_0(dt)$, for all $l=2,\cdots,L+1$;
  \item $\int -t\ln t \tilde{\nu}_1(dt) - \int -t\ln t \tilde{\nu}_0(dt) = 2\eta E_L[-\ln x]_{[\eta,1]}$.
\end{enumerate}
The construction of measures $\tilde{\nu}_0,\tilde{\nu}_1$ are inspired by Wu and Yang~\cite{Wu--Yang2014minimax}.

The following lemma characterizes the properties of $E_L[-\ln x]_{[\eta,1]}$.
\begin{lemma}\label{lem_approx}
  If $K\ge eL^2$, there exists a universal constant $D_0\ge 1$ such that
\begin{align}
  E_L[-\ln x]_{[(D_0K)^{-1},1]} \gtrsim \ln\left(\frac{K}{L^2}\right).
\end{align}
\end{lemma}

Define
\begin{align}\label{eq:lower_minimax_para1}
L = d_2\ln n,  \quad\eta = \frac{nH}{d_2^2D_0S\ln S\ln n}, \quad M = \frac{H}{\ln S}\cdot \frac{d_1}{S\eta} = \frac{d_1d_2^2D_0\ln n}{n},
\end{align}
with universal positive constants $d_1\in(0,e^{-1}],d_2>2$ to be determined later. Without loss of generality we assume that $d_2\ln n$ is always a positive integer. Due to $S\ln S\ge e^2nH\ln n$, we have $(D_0\eta)^{-1}\ge eL^2$, thus Lemma \ref{lem_approx} yields
\begin{align}\label{eq:appro_err}
   E_L[-\ln x]_{[\eta,1]} \gtrsim \ln\left(\frac{1}{D_0\eta L^2}\right) \gtrsim \ln\left(\frac{S\ln S}{Hn\ln n}\right).
\end{align}

Let $g(x)=Mx$ and let $\mu_i$ be the measures on $[0,M]$ defined by $\mu_i(A)=\tilde{\nu}_i(g^{-1}(A))$ for $i=0,1$. It then follows that
\begin{enumerate}
  \item $\int t^1 \mu_1(dt) = \int t^1 \mu_0(dt) = d_1H/(S\ln S)$;
  \item $\int t^l \mu_1(dt) = \int t^l \mu_0(dt)$, for all $l=2,\cdots,L+1$;
  \item $\int -t\ln t \mu_1(dt) - \int -t\ln t \mu_0(dt) = 2\eta M E_L[-\ln x]_{[\eta,1]}$.
\end{enumerate}

Let $\mu_0^{S-1}$ and $\mu_1^{S-1}$ be product priors which we assign to the length-$(S-1)$ vector $(p_1,p_2,\cdots,p_{S-1})$, and we set $p_S=d_1(1-H/\ln S)$. With a little abuse of notation, we still denote the overall product measure by $\mu_0^S$ and $\mu_1^S$. Note that $P$ may not be a probability distribution, we consider the set of \emph{approximate} probability vectors
\begin{align}
  \mathcal{M}_S(\epsilon, H) \triangleq \left\{P:p_1,p_2,\cdots,p_S\ge 0, \left|\sum_{i=1}^S p_i-d_1\right|\le \epsilon, H(P)\le H\right\},
\end{align}
with parameter $\epsilon>0$ to be specified later, and further define under the Poissonized model,
\begin{align}
  R_P(S,n,H,\epsilon) &\triangleq \inf_{\hat{F}}\sup_{P\in \mathcal{M}_S(\epsilon, H)} \mathbb{E}_P|\hat{H}-H(P)|^2.
\end{align}

\begin{lemma}\label{lem_equiv}
  For any $S,n\in\mathbb{N}$ and $0<\epsilon<d_1$, we have
  \begin{align}
      R(S,n,H) &\ge \frac{1}{2d_1^{2}}R_P\left(S,\frac{2n}{d_1},\ln(d_1-\epsilon)\left(H-\ln(d_1+\epsilon)\right),\epsilon\right)\\
       &\qquad - (\ln S)^2\exp(-\frac{n}{4}) - \frac{\epsilon^2}{d_1^{2}}\cdot \sup_{x\in[d_1-\epsilon,d_1+\epsilon]} \ln^2(ex).
  \end{align}
\end{lemma}

In light of Lemma \ref{lem_equiv}, it suffices to consider $R_P(S,n,H,\epsilon)$ to give a lower bound of $R(S,n,H)$. Denote
\begin{align}
  \chi&\triangleq\mathbb{E}_{\mu_1^{S}}H(P) - \mathbb{E}_{\mu_0^{S}}H(P)\\
   &= 2\eta M E_L[-\ln x]_{[\eta,1]}\cdot S\\
   &=\frac{2d_1H}{\ln S}\cdot E_L[-\ln x]_{[\eta,1]}\\
  &\gtrsim \frac{H}{\ln S}\ln\left(\frac{S\ln S}{nH\ln n}\right),
\end{align}
and
\begin{align}
  E_i \triangleq \mathcal{M}_{S}(\epsilon, H)\bigcap \left\{P: |H(P)-\mathbb{E}_{\mu_i^{S}}H(P)|\le\frac{\chi}{4}\right\},\qquad i=0,1.
\end{align}

Denote by $\pi_i$ the conditional distribution defined as
\begin{align}
  \pi_i(A) = \frac{\mu_i^S(E_i\cap A)}{\mu_i^S(E_i)}, \qquad i=0,1.
\end{align}
Now consider $\pi_0,\pi_1$ as two priors. By setting
\begin{align}
  \zeta = \mathbb{E}_{\mu_0^{S}}H(P) + \frac{\chi}{2}, \quad s = \frac{\chi}{4}, \quad \epsilon = \frac{1}{\ln n},
\end{align}
we have $\beta_0=\beta_1=0$ in Lemma \ref{lemma.tsybakov}. Applying union bound yields that
\begin{align}\label{eq:tail_prob}
  \mu_i^{S}[(E_i)^c]&\le \mu_i^{S}\left[\left|\sum_{j=1}^{S} p_j-d_1\right|> \epsilon\right] \\
  &\qquad + \mu_i^{S}\left[|H(P)-\mathbb{E}_{\mu_i^{S}}H(P)|>\frac{\chi}{4}\right] + \mu_i^{S}[H(P)>H]
\end{align}
and the Chebychev inequality tells us that
\begin{align}
  \mu_i^{S}\left[\left|\sum_{j=1}^{S} p_j-d_1\right|> \epsilon\right] &\le \frac{1}{\epsilon^2}\sum_{j=1}^S \mathsf{Var}_{\mu_i^S} (p_j)\\
  &\le \frac{SM^2}{\epsilon^2} \asymp \frac{S(\ln n)^4}{n^2} \\
  &\lesssim \frac{(\ln n)^4}{n^{\frac{1}{2}}} \to 0\label{eq:chebychev_1}\\
  \mu_i^{S}\left[|H(P)-\mathbb{E}_{\mu_i^{S}}H(P)|>\frac{\chi}{4}\right] &\le \frac{16}{\chi^2}\sum_{j=1}^S \mathsf{Var}_{\mu_i^S} (-p_j\ln p_j)\\
 &\le \frac{16S(M\ln M)^2}{\chi^2} \lesssim \frac{S(\ln S)^2(\ln n)^4}{n^2}\\
 & \lesssim \frac{(\ln n)^6}{n^{\frac{1}{2}}} \to 0 \label{eq:chebychev_2}
\end{align}
where we have used our assumption that $S\lesssim n^{\frac{3}{2}}$. For bounding $\mu_i^{S}[H(P)>H]$, we first remark that for $d_1\le e^{-1}$,
\begin{align}
  \bE_{\mu_i^S} H(P) &\le -d_1\ln d_1 + (S-1)\int -t\ln t\mu_i(dt)\\
  &\le -d_1\ln d_1 - S\ln(\eta M)\int t\mu_i(dt)\\
  &= -d_1\ln d_1 + \frac{d_1H}{\ln S}\ln\left(\frac{S\ln S}{d_1H}\right)\\
  &= -d_1\ln d_1 + d_1H - \frac{d_1H}{\ln S}\ln \left(\frac{d_1H}{\ln S}\right)\\
  &\le d_1H - 2d_1\ln d_1
\end{align}
hence, for $d_1$ sufficiently small, say, $d_1\le \min\{\frac{1}{4},f^{-1}\left(\min\{\frac{H_0}{8},\frac{1}{e}\}\right)\}$, where $f(x)=-x\ln x$ is defined in $[0,e^{-1}]$ and $f^{-1}(\cdot)$ denotes the inverse function of $f(\cdot)$, we have
\begin{align}
    \bE_{\mu_i^S} H(P)&\le d_1H - 2d_1\ln d_1\\
    &\le \frac{H}{4} + 2\cdot \min\left\{\frac{H_0}{8},\frac{1}{e}\right\}\\
    & \le \frac{H_0+H}{4}\le \frac{H}{2}.
\end{align}
Hence, similar to (\ref{eq:chebychev_2}), we have
\begin{align}\label{eq:chebychev_3}
  \mu_i^S[H(P)> H]&\le \mu_i^{S}\left[|H(P)-\mathbb{E}_{\mu_i^{S}}H(P)|>\frac{H}{2}\right]\\
 &\le \frac{S(M\ln M)^2}{(H/2)^2}\lesssim \frac{S(\ln n)^4}{n^2}\\
 & \lesssim \frac{(\ln n)^4}{n^{\frac{1}{2}}} \to 0.
\end{align}

Denote by $F_i,G_i$ the marginal probability under prior $\pi_i$ and $\mu_i^S$, respectively, for all $i=0,1$. In light of (\ref{eq:tail_prob}), (\ref{eq:chebychev_1}), (\ref{eq:chebychev_2}) and (\ref{eq:chebychev_3}), we have
\begin{align}
  V(F_i, G_i) \le \mu_i^{S}[(E_i)^c] \to 0.
\end{align}
Moreover, by setting
\begin{align}
  d_1=\min\left\{\frac{1}{4},f^{-1}\left(\min\left\{\frac{H_0}{8},\frac{1}{e}\right\}\right), \frac{1}{d_2^2D_0}\right\}, \quad d_2=10e
\end{align}
it was shown in \cite[Lem. 11]{jiao2014minimax} that
\begin{align}
   V(G_0,G_1) \le \frac{S}{n^6} \lesssim \frac{1}{n^\frac{9}{2}} \to 0.
\end{align}

Hence, the total variational distance is then upper bounded by
\begin{align}
  V(F_0,F_1) &\le V(F_0,G_0) + V(G_0,G_1) + V(G_1,F_1) \to 0
\end{align}
where we have used the triangle inequality of the total variation distance. The idea of converting approximate priors $\mu_i^S$ into priors $\pi_i$ via conditioning comes from Wu and Yang~\cite{Wu--Yang2014minimax}.

Now it follows from Lemma \ref{lemma.tsybakov} and Markov's inequality that
\begin{align}
  R_P(S,n,H,\epsilon) &\ge s^2\inf_{\hat{H}}\sup_{P\in\mathcal{M}_S(\epsilon,H)}\mathbb{P}\left(|\hat{H}-H(P)|\ge s\right)\\
   &\gtrsim \chi^2 \gtrsim \left[ \frac{H}{\ln S}\ln\left(\frac{S\ln S}{nH\ln n}\right)\right]^2
\end{align}
and the desired result follows directly from Lemma \ref{lem_equiv}. Hence we have obtained the desired lower bound in the case $S\lesssim n^{\frac{3}{2}}$.

For $S\gtrsim n^{\frac{3}{2}}$, we can change the parameters in (\ref{eq:lower_minimax_para1}) into
\begin{align}
  L = d_2\ln S,  \quad\eta = \frac{H}{d_2^2D_0S^{\beta}(\ln S)^2}, \quad M = \frac{H}{\ln S}\cdot \frac{d_1}{S\eta} = \frac{d_1d_2^2D_0\ln S}{S^{1-\beta}}
\end{align}
for some small $\beta>0$, say, $\beta=1/4$. Applying the similar analysis yields
\begin{align}
  R(n,S,H) &\gtrsim \left[\frac{H}{\ln S}\ln\left(\frac{S^\beta}{H}\right)\right]^2 \gtrsim H^2\\
  &\asymp \left[\frac{H}{\ln S}\ln\left(\frac{S\ln S}{nH\ln n}\right)\right]^2
\end{align}
which is exactly the desired result.

\section{Future work}

This paper studies the adaptive estimation framework to strengthen the optimality properties of the approximation theoretic entropy estimator proposed in Jiao et al.~\cite{jiao2014minimax}. We remark that the techniques in this paper are by no means constrained to entropy, and we believe analogous results are also true for the estimators of $F_\alpha(P) = \sum_{i = 1}^S p_i^\alpha$ in~\cite{jiao2014minimax}. Furthermore, we find the fact that the sample size enlargement effect still holds in the adaptive estimation setting very intriguing, and we believe there is a larger picture surrounding this theme to be explored.

\section{Acknowledgments}

The authors would like to express their most sincere gratitude to Dany Leviatan for valuable advice on the literature of approximation theory, in particular, for suggesting the result in Lemma~\ref{lem_polynorm}.

\appendix
\section{Auxiliary Lemmas}
The following lemma characterizes the performance of the best uniform approximation polynomial for $-x\ln x, x\in[0,1]$.
\begin{lemma}\label{lem_ibragimov}
  Denote by $\sum_{k=0}^K g_{K,k}x^k$ the $K$-th order best uniform approximation polynomial for $-x\ln x,x\in[0,1]$, then for $p_K(x)=\sum_{k=1}^K g_{K,k}x^k$, we have the norm bound
  \begin{align}\label{eq:bound_norm}
  \sup_{x\in[0,1]} |p_K(x)-(-x\ln x)| \le \frac{D_n}{K^2}
  \end{align}
  where $D_n>0$ is a universal constant for the norm bound. In fact, the following inequality holds:
  \begin{align}
    \limsup_{K\to\infty} K^2\cdot \sup_{x\in[0,1]} |p_K(x)-(-x\ln x)| \le \nu_1(2) \approx 0.453,
  \end{align}
  where the function $\nu_1(p)$ is was introduced by Ibragimov \cite{Ibragimov1946} as the following limit for $p$ positive even integer and $m$ positive integer
  \begin{equation}
\lim_{n\to \infty} \frac{n^p}{(\ln n)^{m-1}} E_n[|x|^p \ln^m |x|]_{[-1,1]} = \nu_1(p).
\end{equation}

Furthermore, we also have the pointwise bound: there exists a universal constant $D_p>0$ such that for any $C\ge1$,
\begin{align}
  |p_K(x) - 2\ln K\cdot x| \le D_pCx, \qquad \forall x\in\left[0,\frac{C}{K^2}\right].
\end{align}
\end{lemma}

\begin{lemma}\label{lem_polynorm}
  \cite[Thm. 8.4.8]{Ditzian--Totik1987} There exists some universal constant $M>0$ such that for any order-$n$ polynomial $p(x)$ in $[0,1]$, we have
  \begin{align}
    \sup_{x\in[0,1]} |p(x)| \le M\cdot\sup_{x\in[n^{-2},1-n^{-2}]} |p(x)|.
  \end{align}
\end{lemma}

The following lemma gives some tails bounds for Poisson and Binomial random variables.
\begin{lemma}\cite[Exercise 4.7]{mitzenmacher2005probability}\label{lemma.poissontail}
If $X\sim \mathsf{Poi}(\lambda)$, or $X\sim \mathsf{B}(n,\frac{\lambda}{n})$, then for any $\delta>0$, we have
\begin{align}
\bP(X \geq (1+\delta) \lambda) & \leq \left( \frac{e^\delta}{(1+\delta)^{1+\delta}} \right)^\lambda \\
\bP(X \leq (1-\delta)\lambda) & \leq  \left( \frac{e^{-\delta}}{(1-\delta)^{1-\delta}} \right)^\lambda \leq e^{-\delta^2 \lambda/2}.
\end{align}
\end{lemma}

\section{Proof of Lemmas}
\subsection{Proof of Lemma \ref{lem_varentropy}}
For $S=1,2$, the result is obvious, and we assume in the sequel that $S\ge 3$. Denote $H(P)=\sum_{i=1}^S -p_i\ln p_i$ by $H$, we construct the Lagrangian:
\begin{align}
  \mathcal{L} = \sum_{i=1}^S p_i(\ln p_i)^2 + \lambda \left(\sum_{i=1}^S -p_i\ln p_i - H\right) + \mu \left(\sum_{i=1}^S p_i-1\right).
\end{align}
By taking the derivative with respect to $p_i$, we obtain that
\begin{align}
  \frac{\partial \mathcal{L}}{\partial p_i} = (\ln p_i)^2 + 2 \ln p_i - \lambda (1+\ln p_i) + \mu
\end{align}
is a quadratic form of $\ln p_i$, so the equation $\frac{\partial \mathcal{L}}{\partial p_i}=0$ has at most two solutions.

Hence, we conclude that components of the maximum achieving distribution can only take two values $p_i\in\{q_1,q_2\}$, and suppose $q_1$ appears $m$ times. We distinguish the analysis into two cases.
\subsubsection{Case I}
If $\min\{q_1,q_2\}\ge \frac{1}{S^2}$, we have $-\ln p_i\le 2\ln S$ for all $i$. Hence,
\begin{align}
  \sum_{i=1}^S p_i(\ln p_i)^2 \le 2\ln S\cdot \sum_{i=1}^S -p_i\ln p_i = 2H\ln S.
\end{align}
\subsubsection{Case II}
If $q_1$ or $q_2$ is smaller than $\frac{1}{S^2}$, without loss of generality we can assume that $q_1\le q_2$ and $q_1< \frac{1}{S^2}$. Then
\begin{align}
  \sum_{i=1}^S p_i(\ln p_i)^2 &= \left(\sum_{i=1}^S -p_i\ln p_i\right)^2 + \sum_{1\le i<j\le S} p_ip_j(\ln p_i-\ln p_j)^2\\
  &= H^2 + m(S-m)q_1q_2(\ln q_1-\ln q_2)^2\\
  &\le H^2 + m(S-m)q_1q_2 (\ln q_1)^2\\
  &\le H^2 + Sq_1 (\ln q_1)^2 \label{eq:appen_basic}\\
  &\le H^2 + S\cdot\frac{1}{S^2}\left(\ln \frac{1}{S^2}\right)^2 \label{eq:appen_monotone}\\
  &= H^2 + \frac{4(\ln S)^2}{S} \\
  &\le H\ln S + 3
\end{align}
where we have used the inequalities $m\le S, (S-m)q_2\le 1$ in (\ref{eq:appen_basic}), and the monotonically increasing property of $x(\ln x)^2$ for $x\in[0, e^{-2}]$ in (\ref{eq:appen_monotone}). The last inequality follows from $H\le \ln S$ and $(\ln S)^2/S \le 4/e^2 < 3/4$.

\subsection{Proof of Lemma \ref{lem_bias}}
  Define $f(x)=-x\ln x$ on $[0,1]$, and the order-$n$ Bernstein polynomial
  \begin{align}
    B_n(f,x) = \sum_{j=0}^n f\left(\frac{j}{n}\right)\cdot \binom{n}{j}x^j(1-x)^{n-j}.
  \end{align}

  It is obvious that $\bE f(\hat{p})=B_n(f,p)$ for $n\hat{p}\sim \mathsf{B}(n,p)$. For $n\ge k$, it follows from \cite[Eqn. (2.3), Chpt. 10]{Devore--Lorentz1993} that
  \begin{align}\label{eq:bernstein_derivative}
    B_n^{(k)}(f,x) = \frac{n!}{(n-k)!}\sum_{j=0}^{n-k}\Delta^k f\left(\frac{j}{n}\right)\cdot \binom{n-k}{j}x^j(1-x)^{n-k-j},
  \end{align}
  where $\Delta^k f(x)$ is the order-$k$ forward difference of $f(\cdot)$ at $x$ with step size $h=1/n$:
  \begin{align}
    \Delta^k f(x) = \sum_{j=0}^k (-1)^{k-j}\binom{k}{j}f\left(x+\frac{j}{n}\right).
  \end{align}

  Since $f^{(3)}(\cdot)\ge 0$, the mean value theorem for the forward difference shows that $\Delta^3 f(\cdot)\ge 0$, and $B_n^{(3)}(f,x)\ge 0$ by (\ref{eq:bernstein_derivative}). Hence $B_n^{(2)}(f,x)$ is monotonically non-decreasing with respect to $x\in[0,1]$, which yields that for $0\le x\le 1/n$,
  \begin{align}\label{eq:bernstein_bound}
    -2\ln2\cdot (n-1) &= B_n^{(2)}(f,0) \le B_n^{(2)}(f,x) \le B_n^{(2)}\left(f,\frac{1}{n}\right)\\
    &\le n(n-1)\cdot \Delta^2 f(0)\left(1-\frac{1}{n}\right)^{n-2} \\
    &\le -\frac{2\ln 2(n-1)}{e}.
  \end{align}

  The proof is completed by applying (\ref{eq:bernstein_bound}) and the Taylor's formula
  \begin{align}
    \bE f(\hat{p}) &= B_n(f,p) = B_n(f,0) + B_n^{(1)}(f,0)p + \frac{1}{2}B_n^{(2)}(f,\xi)p^2\\
    &\in \left[p\ln n- \ln 2(n-1)p^2, p\ln n- \frac{\ln 2}{e}(n-1)p^2\right]
  \end{align}
  for some $\xi\in [0,n^{-1}]$.

\subsection{Proof of Lemma \ref{lem_entropyfunction}}
For $0<x_1\le x_2<1/e$, we write $-y_1\ln y_1=x_1, -y_2\ln y_2=x_2$. It follows from $x_1\le x_2$ that $y_1\le y_2$. By the mean value theorem, 
\begin{align}
  x_2 - x_1 = f(y_2) - f(y_1) = (y_2 - y_1) f'(\xi)
\end{align}
for some $\xi\in[y_1,y_2]$, thus 
\begin{align}\label{eq:meanvalue}
  f^{-1}(x_2) - f^{-1}(x_1) = y_2 - y_1 = \frac{x_2 - x_1}{f'(\xi)} \le \frac{x_2-x_1}{f'(y_2)} = \frac{x_2-x_1}{-\ln y_2-1}.
\end{align}

Now we give an upper bound for $y_2$ in terms of $x_2$. Since $y_2<1/e$, we have
\begin{align}
  y_2 = \frac{x_2}{-\ln y_2} \le \frac{x_2}{-\ln (1/e)} = x_2.
\end{align}
Substituting this result in $-y_2\ln y_2=x_2$ once more yields a refined inequality
\begin{align}\label{eq:upperbound_f}
  y_2 = \frac{x_2}{-\ln y_2} \le \frac{x_2}{-\ln x_2}.
\end{align}

Combining (\ref{eq:meanvalue}) and (\ref{eq:upperbound_f}) gives
\begin{align}\label{eq:f_difference}
  f^{-1}(x_2) - f^{-1}(x_1) \le \frac{x_2-x_1}{-\ln (-x_2/\ln(x_2) )-1}
\end{align}
which completes the proof of the first inequality. For the second equality, we write
\begin{align}
  z = f\left(\frac{1}{y\ln(y\ln y)}\right) = \frac{\ln(y\ln(y\ln y))}{y\ln(y\ln y)}.
\end{align}

It is clear that $z\ge y^{-1}$, thus 
\begin{align}
  f^{-1}(y^{-1})\le f^{-1}(z) = \frac{1}{y\ln(y\ln y)}.
\end{align}
On the other hand, (\ref{eq:f_difference}) asserts that
\begin{align}
  f^{-1}(z) -f^{-1}(y^{-1}) &\le \frac{z-y^{-1}}{-\ln (-z/\ln z)-1}\\
  &\le \frac{1}{\ln(1/z)-1}\cdot \frac{1}{y\ln(y\ln y)}\ln\left(1+\frac{\ln\ln y}{\ln y}\right)\\
  &\le \frac{1}{\ln(1/z)-1}\cdot \frac{\ln\ln y}{y(\ln y)^2}\\
  &= O\left(\frac{\ln\ln y}{y(\ln y)^3}\right).
\end{align}

Combing these two inequalities we have
\begin{align}
  \left|\frac{1}{y\ln(y\ln y)} -f^{-1}(y^{-1})\right| = |f^{-1}(z) -f^{-1}(y^{-1})| = O\left(\frac{\ln\ln y}{y(\ln y)^3}\right)
\end{align}
as desired.

\subsection{Proof of Lemma \ref{lem_small_p}}
For the bias, it is straightforward to see that for $nX\sim\mathsf{Poi}(np)$, we have
\begin{align}
  \bE S_{K,H}(X) + p\ln p &= \sum_{k=1}^K r_{K,H}(4\Delta)^{-k+1}p^k + p\ln p \\
  &= 4\Delta\left[p_K\left(\frac{p}{4\Delta}\right) - \ln(4\Delta)\cdot \frac{p}{4\Delta} \right] + p\ln p\\
  &= 4\Delta\left[p_K\left(\frac{p}{4\Delta}\right) - 2\ln K\cdot \frac{p}{4\Delta} \right] + p\ln\left(\frac{c_2^2n\ln n}{4c_1}\right) + p\ln p
\end{align}
where $p_K(x) \triangleq \sum_{k=1}^K g_{K,H}x^k$ is the best approximating polynomial appearing in Lemma \ref{lem_ibragimov}. Since $\frac{p}{4\Delta}\le \frac{1}{K^2}$, Lemma \ref{lem_ibragimov} asserts that
\begin{align}
 \left|p_K\left(\frac{p}{4\Delta}\right) - 2\ln K\cdot \frac{p}{4\Delta}\right|\le D_p\cdot\frac{p}{4\Delta}
\end{align}
and we conclude that
\begin{align}
  | \bE S_{K,H}(X) + p\ln p| \le -p\ln(pn\ln n) + \left(D_p+\ln(4c_1/c_2^2)\right)p.
\end{align}

The proof for the second part is similar to \cite[Lem. 5]{jiao2014minimax}.

\subsection{Proof of Lemma \ref{lem_approx}}
By defining
\begin{align}
  f_N(x) = -\ln\left(\frac{1+x}{2}+\frac{1-x}{2N}\right),\qquad -1\le x\le 1
\end{align}
we have $E_L[f_N]_{[-1,1]}=E_L[-\ln x]_{[N^{-1},1]}$. Let $\Delta_L(x)=\frac{\sqrt{1-x^2}}{L}+\frac{1}{L^2}$ and define the following modulus of continuity for $f$:
\begin{align}
  \tau_1(f,\Delta_L) \triangleq \sup\left\{|f(x)-f(y)|:x,y\in[-1,1],|x-y|\le \Delta_L(x)\right\}
\end{align}

We have the following lemma.
\begin{lemma}
  There are an upper bound and a lower bound for $\tau_1(f_N,\Delta_L)$:
  \begin{align}
    \ln\left(\frac{N}{2L^2}\right)\le \tau_1(f_N,\Delta_L) \le \ln\left(\frac{2N}{L^2}\right), \qquad \forall L\le \frac{\sqrt{N}}{10}
\end{align}
\end{lemma}
\begin{proof}
  The upper bound is shown in \cite[Lem. 4]{Wu--Yang2014minimax}. For the lower bound, denote by $x_L\in[-1,1]$ the solution to the equation $x_L-\Delta_L(x_L)=-1$, we have the following closed-form formula:
\begin{align}
  x_L = \frac{L^2-L^4+\sqrt{-3L^2+L^4}}{L^2+L^4}\ge -1 + \frac{1}{L^2}.
\end{align}

Hence, by definition, we have
\begin{align}
  \tau_1(f_N,\Delta_L)&\ge |f_N(x_L) - f_N(-1)|\\
&= \ln\left(\frac{x_L+1}{2}N + \frac{1-x_L}{2}\right)\\
&\ge \ln\left(\frac{x_L+1}{2}N\right)\\
&\ge \ln\left(\frac{N}{2L^2}\right).
\end{align}
\end{proof}

The relationship between $\tau_1(f_N,\Delta_L)$ and $E_L[f_N]_{[-1,1]}$ was shown in \cite[Thm. 3.13, Thm. 3.14]{Petrushev--Popov2011rational} that there exist two universal constants $M_1,M_2>0$ such that
\begin{align}
  E_n[f_N]_{[-1,1]} &\le M_1\tau_1(f_N,\Delta_n) \label{eq:approx_upper}\\
  \frac{1}{n}\sum_{k=0}^n E_k[f_N]_{[-1,1]} &\ge M_2\tau_1(f_N,\Delta_n) \label{eq:approx_lower}
\end{align}

Applying (\ref{eq:approx_upper}) and (\ref{eq:approx_lower}) and setting the approximation order to be $DL$ with constant $D>1$ to be specified later, then given $N=(10D)^2M\ge (10DL)^2$, the non-increasing property of $E_n[f_N]_{[-1,1]}$ with respect to $n$ yields
\begin{align}
  E_L[f_N]_{[-1,1]} &\ge \frac{1}{DL-L}\sum_{n=L+1}^{DL} E_n[f_N]_{[-1,1]}\\
  &\ge \frac{1}{DL}\left(\sum_{n=0}^{DL}E_n[f_N]_{[-1,1]} - E_0[f_N]_{[-1,1]} - \sum_{n=1}^L E_n[f_N]_{[-1,1]}\right)\\
  &\ge M_2\tau_1(f_N,\Delta_{DL}) - \frac{\ln N}{DL} - \frac{M_1}{DL}\sum_{n=1}^L \tau_1(f_N,\Delta_n)\\
  &\ge M_2\ln\left(\frac{N}{2(DL)^2}\right) - \frac{\ln N}{DL} - \frac{M_1}{DL}\sum_{n=1}^L \ln\left(\frac{2N}{n^2}\right)\\
  &\ge M_2\ln\left(\frac{N}{2(DL)^2}\right) - \frac{\ln N}{DL} - \frac{M_1}{DL}\int_1^L \ln\left(\frac{2N}{x^2}\right)dx\\
  &\ge M_2\ln\left(\frac{50K}{L^2}\right) - \frac{\ln K+2\ln(10D)}{DL} - \frac{M_1}{D}\ln\left(\frac{200e^2D^2K}{L^2}\right).
\end{align}

Hence, there exists a sufficiently large constant $D>0$ such that
\begin{align}
  E_L[-\ln x]_{[(100D^2K)^{-1},1]} = E_L[f_N]_{[-1,1]} \gtrsim \ln\left(\frac{K}{L^2}\right)
\end{align}
and this lemma is proved by setting $D_0=\max\{100D^2,1\}$.

\subsection{Proof of Lemma \ref{lem_equiv}}
Fix $\delta>0$. Let $\hat{H}({\bf Z})$ be a near-minimax estimator of $H(P)$ under the Multinomial model. The estimator $\hat{H}({\bf Z})$ obtains the number of samples $n$ from observation $\bf Z$. By definition, we have
  \begin{align}
    \sup_{P\in \mathcal{M}_S(H)} \mathbb{E}_{\mathrm{Multinomial}}|\hat{H}(\mathbf{Z})-H(P)|^2 < R(S,n,H)+\delta,
  \end{align}
  where $R(S,n,H)$ is the minimax $L_2$ risk under the Multinomial model. Note that for any vector $P\in\mathcal{M}_S(\epsilon,H)$ ($P$ is not necessarily a probability distribution), we have
\begin{align}\label{eq:entropy_rela}
  H\left(\frac{P}{\sum_{i=1}^S p_i}\right) = \frac{H(P)}{\sum_{i=1}^S p_i} + \ln\left(\sum_{i=1}^S p_i\right) \le \frac{H(P)}{d_1-\epsilon} + \ln(d_1+\epsilon)
\end{align}
where by definition we have $\left|\sum_{i=1}^S p_i-d_1\right|\le \epsilon$. Hence, given $P\in\mathcal{M}_S(\epsilon, H)$, let $\mathbf{Z}=[Z_1,\cdots,Z_S]^T$ with $Z_i\sim \mathsf{Poi}(np_i)$ and let $n'=\sum_{i=1}^S Z_i\sim \mathsf{Poi}(n\sum_{i=1}^Sp_i)$, (\ref{eq:entropy_rela}) suggests to use the estimator $d_1\left(\hat{H}(\mathbf{Z})-\ln d_1\right)$ to estimate $H(P)$. Note that
\begin{align}
  d_1\left(\hat{H}(\mathbf{Z})-\ln d_1\right)-H(P) = d_1\left(\hat{H}(\mathbf{Z}) - H\left(\frac{P}{\sum_{i=1}^S p_i}\right)\right)
  + \left(\sum_{i=1}^S p_i\right)\ln\left(\sum_{i=1}^S p_i\right) - d_1\ln d_1
\end{align}
the triangle inequality gives (define $A=\sup_{x\in[d_1-\epsilon,d_1+\epsilon]} \ln^2(ex)$)
  \begin{align}
    &\quad \frac{1}{2}\mathbb{E}_P\left|d_1\left(\hat{H}(\mathbf{Z})-\ln d_1\right)-H(P)\right|^2\\
&\le d_1^2\mathbb{E}_P\left|\hat{H}({\bf Z})-H\left(\frac{P}{\sum_{i=1}^S p_i}\right)\right|^2 + \left|\left(\sum_{i=1}^S p_i\right)\ln\left(\sum_{i=1}^S p_i\right) - d_1\ln d_1\right|^2\\
    &\le d_1^{2} \sum_{m = 0}^\infty \mathbb{E}_P\left[\left.\left|\hat{H}(\mathbf{Z})-H\left(\frac{P}{\sum_{i=1}^S p_i}\right)\right|^2\right|n'=m\right]\mathbb{P}(n'=m)+ \epsilon^2A\\
    &\le d_1^{2}\sum_{m=0}^\infty R\left(S,m,\frac{H}{d_1-\epsilon} + \ln(d_1+\epsilon)\right)\mathbb{P}(n'=m) + \delta + \epsilon^2A\\
    &\le d_1^{2}R\left(S,\frac{d_1n}{2},\frac{H}{d_1-\epsilon} + \ln(d_1+\epsilon)\right)\mathbb{P}(n'\ge \frac{d_1n}{2}) + (d_1\ln S)^2\mathbb{P}(n'\le \frac{d_1n}{2}) + \delta +\epsilon^2A\\
    &\le d_1^{2}R\left(S,\frac{d_1n}{2},\frac{H}{d_1-\epsilon} + \ln(d_1+\epsilon)\right)+ (d_1\ln S)^2\exp(-\frac{d_1n}{8}) + \delta + \epsilon^2A,
  \end{align}
  where we have used the fact that conditioned on $n'=m$, $\mathbf{Z}\sim \mathsf{Multinomial}(m,\frac{P}{\sum_ip_i})$, and $R(S,n,H)\le \left(\sup_{P\in\mathcal{M}_S} H(P)\right)^2=(\ln S)^2$. Moreover, the last step follows from Lemma~\ref{lemma.poissontail}. The proof is completed by the arbitrariness of $\delta$ and Lemma~\ref{lemma.poissonmultinomial}.

\subsection{Proof of Lemma \ref{lem_ibragimov}}
It has been shown in \cite[Lemma 18]{jiao2014minimax} that
\begin{align}
  \lim_{K\to\infty} K^2\cdot \sup_{x\in[0,1]}\left|\sum_{k=0}^K g_{K,k}x^k-(-x\ln x)\right| = \frac{\nu_1(2)}{2},
\end{align}
then plugging in $x=0$ yields
\begin{align}
  \limsup_{K\to\infty} K^2\cdot |g_{K,0}| \le \frac{\nu_1(2)}{2}.
\end{align}
Hence, it follows from the triangle inequality that
\begin{align}
  \limsup_{K\to\infty} K^2\cdot \sup_{x\in[0,1]}\left|\sum_{k=1}^K g_{K,k}x^k-(-x\ln x)\right| \le \frac{\nu_1(2)}{2} + \frac{\nu_1(2)}{2} = \nu_1(2),
\end{align}
which completes the proof of the norm bound.

For the pointwise bound, \cite[Thm. 7.3.1]{Ditzian--Totik1987} asserts that there exists a universal positive constant $M_1$ such that
\begin{align}
  \sup_{x\in[0,1]} |(\varphi(x))^2 p_K''(x)| \le M_1K^2\omega_\varphi^2(-x\ln x,K^{-1}),
\end{align}
where $\varphi(x)=\sqrt{x(1-x)}$, and $\omega_\varphi^2(f,t)$ is the second-order Ditzian-Totik modulus of smoothness \cite{Ditzian--Totik1987} defined by
\begin{align}
  \omega_\varphi^2(f,t)\triangleq \sup\left\{\left|f(u)+f(v)-2f\left(\frac{u+v}{2}\right)\right|: u,v\in[0,1], |u-v|\le 2t\varphi\left(\frac{u+v}{2}\right)\right\}.
\end{align}

Direct computation yields
\begin{align}
  \omega_\varphi^2(-x\ln x,t) = \frac{2t^2\ln 2}{1+t^2},
\end{align}
we have
\begin{align}
  \sup_{x\in[0,1]} |x(1-x) p_K''(x)|\le 2M_1\ln 2.
\end{align}

According to Lemma \ref{lem_polynorm}, since $p_K''(x)$ is a polynomial with order $K-2<2K$, there exists some positive constant $M_2$ such that
\begin{align}
  \sup_{x\in[0,1]} |p_K''(x)|
  &\le M_2\sup_{x\in[(2K)^{-2},1-(2K)^{-2}]} |p_K''(x)| \\
  &\le \frac{M_2(2K)^4}{(2K)^2-1}\sup_{x\in[(2K)^{-2},1-(2K)^{-2}]} |x(1-x)p_K''(x)| \\
    &\le \frac{16M_2K^4}{4K^2-1}\sup_{x\in[0,1]} |x(1-x)p_K''(x)| \\
    &\le \frac{32M_1M_2K^4\ln 2}{4K^2-1}\\
    &\le 16M_1M_2K^2\ln2,
\end{align}
hence for any $x,y\in[0,C/K^2]$, we have
\begin{align}
  |p_K'(x)-p_K'(y)| &\le \int_{\min\{x,y\}}^{\max\{x,y\}} |p_K''(t)|dt \\
  &\le 16M_1M_2\ln2\cdot K^2|x-y|\\
  &\le 16M_1M_2C\ln2.
\end{align}

As a result, we know that for any $C\ge 1$, $u=C/K^2$ and $x\in[0,u]$,
\begin{align}
  16M_1M_2C\ln 2&\ge \frac{1}{u}\int_0^{u}\left|p_K'(x)-p_K'(t)\right|dt\\
  &\ge \left|p_K'(x)-\frac{1}{u}\int_0^{u}p_K'(t)\right|dt\\
  &= \left|p_K'(x) - \frac{p_K(u)}{u}\right|\\
  &\ge \left|p_K'(x) + \ln u\right| - \frac{\left|p_K(u)-(-u\ln u)\right|}{u}\\
  &\ge \left|p_K'(x) + \ln u\right| - K^2\sup_{t\in[0,1]} \left|p_K(t)-(-t\ln t)\right|\\
  &\ge \left|p_K'(x) - 2\ln K\right| - \ln C - D_n,
\end{align}
where $D_n$ is the coefficient of the norm bound in (\ref{eq:bound_norm}). Hence, the universal positive constant $D_p\triangleq 16M_1M_2\ln2+1+D_n$ satisfies
\begin{align}
  \left|p_K'(x) - 2\ln K\right| \le D_pC, \qquad \forall x\in\left[0,\frac{C}{K^2}\right],
\end{align}
and it follows that
\begin{align}
  \left|p_K(x) - 2\ln K\cdot x\right| &\le \int_0^x \left|p_K'(t) - 2\ln K\right|dt \\
  &\le \int_0^x D_pCdt = D_pCx.
\end{align}

\bibliographystyle{IEEEtran}
\bibliography{references,di}

\end{document}